\documentclass[11pt]{article}
\usepackage{fullpage}
\usepackage{amsmath,amssymb, amsthm}
\usepackage[utf8]{inputenc}
\usepackage{bbm,color,tikz}
\usetikzlibrary{arrows.meta}
\usetikzlibrary{math}
\usepackage[ruled,vlined]{algorithm2e}
\usepackage{parskip, nicefrac}

\title{Interactive Communication in Bilateral Trade}
\author{Jieming Mao \\ Google Research \and Renato Paes Leme \\ Google Research \and Kangning Wang \\ Duke University}

\date{}

\usepackage{url}

\newcommand{\D}{\mathcal{D}}

\newcommand{\MC}{\mathbf{C}}
\newcommand{\VR}{\mathbf{VR}}

\newcommand{\E}{\mathbb{E}}
\newcommand{\R}{\mathbb{R}}

\newcommand{\X}{\mathcal{X}}

\newcommand{\one}{\mathbf{1}}
\newcommand{\arw}{{Latex[length=1.5mm,width=2mm]}}
\newcommand\abs[1]{\vert{#1}\vert}

\newtheorem{theorem}{Theorem}[section]
\newtheorem{lemma}[theorem]{Lemma}

\theoremstyle{definition}

\newtheorem{example}{Example}

\begin{document}

\maketitle

\begin{abstract}
We define a model of interactive communication where two agents with private types can exchange information before a game is played. The model contains Bayesian persuasion as a special case of a one-round communication protocol. We define message complexity corresponding to the minimum number of interactive rounds necessary to achieve the best possible outcome. Our main result is that for bilateral trade, agents don't stop talking until they reach an efficient outcome: Either agents achieve an efficient allocation in finitely many rounds of communication; or the optimal communication protocol has infinite number of rounds. We show an important class of bilateral trade settings where efficient allocation is achievable with a small number of rounds of communication.
\end{abstract}

\def\arraystretch{1.5}


\section{Introduction}

We consider the situation where agents are allowed to have a conversation before playing a game. Unlike cheap talk \cite{aumann2003long}, we assume agents are able to send verifiable signals. The ability to send verifiable signals is a new device introduced in the literature on Bayesian persuasion \cite{kamenica2011bayesian}. Without this power, various revelation-principle-type results show that anything that can be achieved through an interactive game/mechanism can also be achieved through direct revelation. In stark contrast to those, we will show that with verifiable signals, a (potentially very long) interactive protocol can lead to more efficient outcomes.

Consider the following scenario: Sally is a supplier of electronic components and Bob is a builder of computers. They have an ongoing business relationship and often negotiate the price of new custom parts. Since they come to the negotiation table over and over, they established a protocol which they follow during the negotiation. Through this protocol, Bob will convey information about his value for the new component and Sally will convey information about her cost for producing it. Each step of the communication is verifiable: Bob can for example show Sally a quote from a competing supplier, bounding his willingness to pay. Sally can tell Bob the cost of raw materials for building that component, which will refine Bob's information about Sally's production cost. Their goal is that if they talk long enough, they will be able to settle on a price in-between or figure out that no trade is possible (Sally's cost exceeds Bob's value).

Informally, we will show that the optimal communication protocol will  only end when agents reach an efficient outcome. To formalize this statement, we define a formal model of interactive communication, which contains the Bayesian persuasion model of  Kamenica and Gentzkow \cite{kamenica2011bayesian} and Dughmi and Xu \cite{dughmi2016persuasion} as a special case of a one-round protocol. We define our protocol for a generic two-player game and study bilateral trade as a special case. Although the model can be generalized to multiple players, most applications we have in mind are about two agents talking to each other: buyer and seller in bilateral trade, sender and receiver in Bayesian persuasion or Alice and Bob in standard communication complexity.

\paragraph{Model of Interactive Communication} In this paper we study this question by first defining a communication protocol that takes place before Sally and Bob play a game where both players have private information and their payoffs depend on the combination of their private types. To make this concrete, we will use \emph{bilateral trade} as the main running example. Sally (the seller) has a private cost $\theta_S$ on the item being sold and Bob (the buyer) has a private value $\theta_B$. Sally's action in the game is to set a price $p$. If the price is below Bob's value the trade happens and Sally and Bob get utilities $p-\theta_S$ and $\theta_B - p$ respectively. Otherwise trade doesn't happen and both get zero utility.

Before Sally decides on the price, Bob may want to send a signal refining Sally's information in an attempt to persuade her to set a lower price. Since this is a game where both sides have private information, Bob's signal will depend on his information about Sally. Before Bob sends a signal, Sally may want to refine Bob's information about her in an attempt to persuade him to send a better signal.
This can be taken one step further. Before Sally sends her signal, Bob can send a preliminary signal to influence how Sally will signal to influence Bob's subsequent signal.

The result can be interpreted as a communication protocol which consists of an alternating sequence of information refinements. In odd rounds, Bob sends a signal refining Sally's information and in even rounds Sally sends a signal refining Bob's information. We assume the same commitment structure as in Kamenica and Gentzkow \cite{kamenica2011bayesian}: agents truthfully communicate the realization of signals. This leads to a natural equilibrium notion which contains the persuasion scheme of \cite{kamenica2011bayesian} as a special case when the round complexity of the protocol is one.

To study how the round complexity of the protocol affects the efficiency of the game, we define the notion of \emph{message complexity}. A game has message complexity $t$ if players can't improve their payoffs after $t$ rounds of communication. In other words, for any $t' \geq t$ the payoffs the agents get with a $t'$-round protocol are the same as with a $t$-round protocol.

\paragraph{Efficient Communication in Bilateral Trade} Our main result is that for bilateral trade, if the message complexity is finite, then the allocation is efficient. This means that for any initial set of types and distributions in a bilateral trade game, one of the following must be true: either agents exchange information in a way that after finitely many rounds they can implement the efficient outcome (which is to trade whenever $\theta_B > \theta_S$); or the efficiency of the allocation keeps improving the longer they talk.

For the case where the buyer has only two types in the support of his distribution, efficient allocation can always be achieved with two rounds of communication (regardless of how complicated the seller's type space is). If the buyer has three types in the support, we have an example showing we need at least three rounds of communication. Finally, we conjecture that whenever type spaces are finite we can always achieve efficient allocation in finitely many rounds of communication.

\paragraph{A game with longer communication} We end the paper with an example of a game with longer communication. In the tradition of Cold War game theory we now have Sally (the Spymaster) and Bob (the Birdwatcher) be spies whose private types correspond to the countries they are actually serving. Sally and Bob don't know if they are friends or enemies. If they are friends (same type) their payoffs correspond to those of a cooperative game. If they are enemies (different types) their payoffs correspond to a zero-sum game. Bob and Sally can communicate before playing the game, but it is a delicate balance: both would like to know if their counterpart is their friend or enemy but both would also like to the other to believe they may be a friend when they are in fact enemies. We show that a long and gradual disclose of information can  benefit both parties.

\paragraph{Techniques} From a technical standpoint, we build on top of two techniques introduced by 
Aumann and Maschler  \cite{aumann1995repeated} that became standard in the toolbox of information design. The first is the observation that signals can be thought of as decomposition of a prior distribution into posteriors that are only required to average to the prior (``splitting lemma''). The second is the concavification of payoff functions, which is the observation that by signaling, the sender can replace their payoff function by its concave hull.

Our main technical innovation is to analyze the dynamic of \emph{alternate concavification}, by which we mean the following: consider we start from payoff functions $\pi_B(\D_S, \D_B)$ and $\pi_S(\D_S, \D_B)$ that are defined as functions of Sally's information about Bob ($\D_B$) and Bob's information about Sally ($\D_S$). A message by Sally allows her to replace her payoff by its concave hull with respect to $\D_S$, smoothening out Bob's payoff correspondingly. Bob can similarly replace his payoff by its concave hull on $\D_B$. Note however, that whenever Sally concavifies her payoff it may cause Bob's payoff to be no longer concave and vice-versa. Our analysis will deal with understanding whether this procedure stabilizes after finitely many rounds (finite message complexity) or whether it goes on indefinitely. If it stabilizes, we are interested in understanding its properties.

In our analysis of bilateral trade, we will develop a higher-order version of the indifference argument of Bergemann, Brooks and Morris \cite{bergemann2015limits}.
While \cite{bergemann2015limits} decomposes a prior distribution into posteriors for which the seller is indifferent about which price to set, we will decompose each prior into posteriors for which the buyer is indifferent between different ways to make the seller indifferent.


\paragraph{Implementation}  In an online supplement\footnote{See code in Github (\url{https://gist.github.com/renatoppl/72ee85d212d08ad9670977cc8ffa2afa}) or Google Colab (\url{https://colab.research.google.com/drive/1lQdRrZD7-aCuuYJM6bI1sf1fi_CSA4Fq?usp=sharing}).} we provide an implementation of the alternating concavification procedure which takes a specification of a generic game with binary type spaces as an input and computes the payoffs after $t$ rounds of communication together with the communication protocol. The computation is exact: it uses rational numbers (so there is no floating point precision issues) and computes a parametric concave hull so it doesn't need to rely on discretization.

\paragraph{Related Work} Our work is broadly situated in line of work in Economics studying how the information structure affects the outcome in auction and bargaining settings, which was initiated in Bergemann and Pesendorfer \cite{bergemann2007information} and Es{\H{o}} and Szentes \cite{esHo2007optimal} and has been more recently explored in Bergemann et al \cite{bergemann2007information,bergemann2017first}, Emek et al \cite{emek2014signaling}, Daskalakis et al \cite{daskalakis2016does} and Badanidiyuru et al \cite{badanidiyuru2018targeting}. In this line of work, the auction designer is more informed than the participants and must decide how much information to disclose as part of the design decision. We differ from this line of work in the sense that we assume that the information lies with the participants themselves and their decisions on how to disclose information affect their payoffs.

In that sense, our work is closer to the Bayesian persuasion model of Kamenica and Gentzkow~\cite{kamenica2011bayesian,gentzkow2014costly} and the algorithmic persuasion of Dughmi and Xu \cite{dughmi2019algorithmic} and Dughmi et al \cite{dughmi2016persuasion}. In this setting there are two agents: \emph{sender} and \emph{receiver} where the sender is typically more informed than the receiver but the receiver is the one responsible for choosing an action in a base game. In contrast our agents both have uncertainty about each other. Kamenica and Gentzkow \cite{kamenica2011bayesian} briefly consider the setting where the receiver can also have private information, but only allow messages in one direction (from the sender to the receiver). Instead our paper considers both parties to have partial information and considers an interactive exchange of information. Doval and Ely \cite{doval2020sequential} consider sequential disclosure of information that may depend on player behaviors, but it again differs from our setting in the sense that the information is initially held by the designer who then discloses it to agents. We refer to the excellent survey by Bergemann and Morris \cite{bergemann2019information} for a unified treatment of those papers.

The notion of having a conversation before a game is played is the central talk in the \emph{cheap talk} literature. In particular in ``Long Cheap Talk'', Aumann and Hart \cite{aumann2003long} show how the set of equilibrium payoffs of a game can be expanded by an arbitrarily long conversation before the game is played. Their model assumes no commitment whatsoever (the players are free to send messages as they please) while in our model we have the same commitment structure as in the Bayesian persuasion literature, where the realization of signals are truthfully communicated.

The power of interactive communication has been extensively studied in communication complexity (Yao \cite{Yao79}). Nisan and Wigderson \cite{NisanW93} show an exponential gap between $k$ and $(k-1)$-round communication complexity. Babai et al \cite{BabaiGKL03} show an exponential gap between simultaneous communication complexity and communication complexity in the multi-party number-on-forehead communication model. Similar demonstrations of the power of interactive communication have been shown when studying the communication complexity of mechanism design problems. Dobzinski, Nisan and Oren~\cite{DobzinskiNO14} study how the number of adaptive rounds affects welfare efficiency in communication protocols with polynomial communication complexity for combinatorial auctions. Subsequent work by Alon et al \cite{AlonNRW15} and Assadi \cite{Assadi17} provides tight bounds on the number of rounds necessary to obtain an efficient allocation. The motivation for interactive communication in those papers is a restriction in the number of bits used in each interaction. Hence the need for interactivity comes from algorithmic and not strategic considerations. In our paper, on the other hand, interactivity is driven by strategic considerations: agents will only reveal so much about their types until they can learn more about the other agent's type.

Our results also contribute to the line of work on bilateral trade started by Chatterjee and and Samuelson \cite{chatterjee1983bargaining} and Myerson and Satterthwaite \cite{myerson1983efficient}. Their message is in a sense the opposite of ours: interaction is not helpful (in the sense of the revelation principle) and efficient trade can't be achieved by an incentive compatible mechanism. We show that communication is useful and leads to efficient allocation in many important cases. The main difference is that here we are giving the agents additional commitment power: they can credibly signal about their type, which is what drives the Bayesian persuasion literature. We refer the reader to Section 1C of Kamenica and Gentzkow \cite{kamenica2011bayesian} for an in-depth discussion of the source and motivation behind this additional commitment power.

There is an important line of work looking at the bilateral trade problem from the perspective of approximation algorithms. Blumrosen and Dobzinski \cite{blumrosen2014reallocation} give the first approximation to efficiency in bilateral trade, later improved by Collini-Baldeschi et al \cite{colini2016approximately} and Kang and Vondrak~\cite{kang2019fixed}. Recently, various new angles have been explored, such as multi-dimensional two-sided markets (Collini-Baldeschi et al \cite{colini2020approximately}, Cai et al \cite{cai2021multi}), gains from trade approximation (Brustle et al \cite{brustle2017approximating}
), best of both worlds guarantees (Babaioff et al \cite{babaioff2018best}), resource augmentation (Babaioff et al \cite{babaioff2020bulow}), dynamic auctions (Balseiro et al \cite{balseiro2019dynamic}), and limited information (D{\"u}tting et al \cite{dutting2020efficient}).


\section{Games with Interactive Communication}\label{sec:communication}

Our main objects of study are two-player games where both players (called Sally and Bob) have private types. Ex-ante (before types are revealed) the players can agree on a communication protocol to exchange information about their types. We will be interested in how the round complexity of the communication protocol can affect the outcome of the game. For simplicity we will restrict to games where only Sally has non-trivial actions.

\subsection{Base game} We first define the \emph{base game}, which is played after the communication protocol. In this game Bob has a type $\theta_B \in \Theta_B$ and Sally has a type $\theta_S \in \Theta_S$. The types are drawn from independent known distributions $\D_B$ and $\D_S$ respectively. Only Sally has an action $a_S \in A_S$ to choose. The utilities of both players are given by functions:
$$u_i : A_S \times \Theta_S \times \Theta_B \rightarrow \R \quad i \in \{S, B\}$$
As usual in games with private information, Sally knows her type but only knows the distribution over which Bob's type is drawn and vice versa. Her optimal strategy is rather simple:
$$a_S^*(\theta_S;\D_B) = \text{argmax}_{a \in A_S} \E_{\theta_B \sim \D_B}[u_S(a,\theta_S, \theta_B)]$$

\subsection{Example: Bilateral Trade} \label{ex:bilateral_trade}
Bilateral trade will provide us the main running example. Sally the seller is trying to sell an item to Bob the buyer. The types of Sally and Bob are their values for the item. Sally's action is to choose a price to sell the good. The sets $\Theta_S$, $\Theta_B$ and $A_S$ correspond to the non-negative real numbers $\R_+$. Upon setting a price $a_S$, Bob will buy whenever $\theta_B \geq a_S$ leading to the following utilities:
$$u_B(a_S, \theta_S, \theta_B) = (\theta_B - a_S) \cdot \one\{ \theta_B \geq a_S\} \qquad u_S(a_S, \theta_S, \theta_B) = (a_S - \theta_S) \cdot \one\{ \theta_B \geq a_S\}$$
Without any communication, the expected welfare is:
$$W (\D_S, \D_B) = \E_{\theta_S \sim \D_S} \E_{\theta_B \sim \D_B} [u_S(a_S^*(\theta_S; \D_B), \theta_S, \theta_B) + u_B(a_S^*(\theta_S; \D_B), \theta_S, \theta_B)]$$
which is typically suboptimal when compared to the the welfare under efficient trade:
\begin{equation}\label{eq:eff_welfare}
W ^*(\D_S, \D_B) = \E_{\theta_S \sim \D_S} \E_{\theta_B \sim \D_B} [(\theta_B - \theta_S)^+]
\end{equation}

If Bob doesn't buy, both players experience zero utility. So it is in the best interest of both Sally and Bob to make the trade happen whenever $\theta_B > \theta_S$. It is only natural for them to engage in \emph{bargaining}, i.e., a conversation in which information is gradually revealed in a way that the price settles in at a point that hopefully leads to efficient trade. By gradual information revealing we mean gradually refining each other's priors to lead to an outcome that is improving for both parties. 

\subsection{Communication Protocol}

A communication protocol will be an alternating sequence of information refinements: in odd rounds Bob will send a signal to refine Sally's information about his type and in even rounds Sally will refine Bob's information about her type.

\subsubsection{Notation}
The protocol will be described as a sequence of messages by each player. We will adopt the unusual but convenient convention of using $t$ to denote the $t$-to-last message. So $m^1$ will denote the last message, $m^2$ the penultimate message and so on$\ldots$ As we described before, odd messages will be sent by Bob and denoted $m^t_B$. Even messages will be sent by Sally and denoted $m^t_S$.

To describe the messages, we need an additional notation. Given any set $\X$, let $\Delta(\X)$ be the set of distributions over $\X$, and let $\Delta(\Delta(\X))$ be the set of distributions over $\Delta(\X)$. 
We will also define the following operator:
$$\mu : \Delta(\Delta(\X)) \rightarrow \Delta(\X)$$
that given a distribution over distributions returns a single distribution in the natural way: given $M \in \Delta(\Delta(\X))$ build a distribution over $\X$ by first sampling $\D \sim M$ where $\D \in \Delta(\X)$ and then sampling an element of $\X$ from $\D$. The resulting distribution over $\X$ is $\mu(M)$.

\subsubsection{Information Refinement}\label{subsec:refinement}

We say that a distribution over distributions $M \in \Delta(\Delta(\X))$ is an information refinement of a distribution $\D \in \Delta(\X)$ whenever $\mu(M) = \D$.

An information refinement is a convenient way to represent a signal. A signal about $\D$ is a random variable $Y$ that is correlated with a random variable $X \sim \D$. Upon observing $Y$ one can perform a Bayesian update and obtain the distribution of $X$ conditioned on $Y$. This induces a distribution over distributions $M \in \Delta(\Delta(\X))$ s.t. $\mu(M)=\D$. Conversely given any $M \in \Delta(\Delta(\X))$ such that $\mu(M)=\D$ we can obtain a signal by letting $Y$ represent the distribution sampled from $M$ and $X$ be an element in $\X$ sampled from $Y$ (here $Y$ is both a random variable and a distribution in $\Delta(\X)$).

\subsubsection{Structure of the messages}
\label{subsubsec:msg_struct}
Using the notation in the previous subsections we define messages in odd rounds as:
$$m^t_B: \Delta(\Theta_S) \times \Delta(\Theta_B) \rightarrow  \Delta(\Delta(\Theta_B)) \quad \text{s.t.} \quad \mu(m^t_B(\D_S, \D_B)) = \D_B$$
A message will take as input the information Sally has about Bob and the information Bob has about Sally and output a refined version of Bob's information. Similarly in even rounds:
$$m^t_S: \Delta(\Theta_S) \times \Delta(\Theta_B) \rightarrow  \Delta(\Delta(\Theta_S)) \quad \text{s.t.} \quad \mu(m^t_S(\D_S, \D_B)) = \D_S$$
A sequence of functions $m_B^1, m_S^2, m_B^3, \hdots, m_i^k$ (where $i \in \{B,S\}$ depends on the parity of $k$) describes an \emph{alternating information refinement} protocol.

The protocol is executed then as follows: if the types are sampled from $\D_B$ and $\D_S$ respectively, we set $\D_B^k = \D_B$ and $\D_S^k = \D_S$. Then for $t=k,k-1,k-2, \hdots, 1$,
$$\begin{aligned}
& \text{if $t$ is odd:} & & \D_B^{t-1} \sim m_B^t(\D_S^t, \D_B^t) & \quad & \D_S^{t-1} = \D_S^{t}\\
& \text{if $t$ is even:} & & 
\D_B^{t-1} = \D_B^t  & & \D_S^{t-1} \sim m_S^t(\D_S^t, \D_B^t)
\end{aligned}$$
Once the communication is over, Sally selects her optimal action $a_S^*(\theta_S; \D_B^0)$ using the information available at that point.

\subsubsection{Payoffs}

We can now define recursively the payoffs of each agent. Let $\pi_B^t(\D_S, \D_B)$ and $\pi_S^t(\D_S, \D_B)$ be the expected payoffs obtained by Bob and Sally from a $t$-round protocol with initial information $\D_S$ and $\D_B$. For $t=0$ we simply have the payoffs of the base game:
\begin{equation}\label{eq:pi0_def}
\pi^0_i(\D_S, \D_B) = \E_{\theta_S \sim \D_S, \theta_B \sim \D_B} [u_i(a_S^*(\theta_S; \D_B), \theta_S, \theta_B)] \quad i \in \{B, S\}
\end{equation}
For $t > 0$ and we have:
$$\begin{aligned}
&
\pi^t_i(\D_S, \D_B) = \E_{\D'_B \sim m_B^t(\D_S, \D_B)} [\pi_i^{t-1}(\D_S, \D'_B)] \quad & & i \in \{B, S\} \quad & & \text{for odd } t \\
& \pi^t_i(\D_S, \D_B) = \E_{\D'_S \sim m_S^t(\D_S, \D_B)} [\pi_i^{t-1}(\D'_S, \D_B)] \quad & & i \in \{B, S\} \quad & & \text{for even } t \\
\end{aligned}$$

\subsubsection{Solution Concept}

We now define the notion of equilibrium of a communication protocol. In high level terms a protocol in equilibrium must satisfy two properties: (i) \emph{voluntary communication} and (ii) \emph{sub-protocol optimality}.

\paragraph{Voluntary Communication}
We say that a protocol satisfies \emph{voluntary communication} if both agents weakly prefer communicating over skipping that round. Mathematically this can be stated as follows:
\begin{equation}\label{eq:voluntary_refinement}
\pi^t_i(\D_S, \D_B) \geq \pi^{t-1}_i(\D_S, \D_B), \forall i \in \{B, S\} \text{ and } \forall t > 0
\end{equation}
This means that the sender will never send a message decreasing their payoff. The receiver will refuse\footnote{The important assumption here is that the receiver has a way to credibly ``not listen''. For example, whenever Bob tells Sally something, even if she decides not to use the information in subsequent rounds, Bob knows that she knows and that may prevent Bob from disclosing further information in the future. Sally must have a way to show Bob that she didn't update her prior based on that message.} to hear any message that decreases their payoff (by for example shutting their ears or deleting an email without reading).

We define the set of \emph{voluntary refinements} as follows. For odd $t$ define $\VR_B^t(\D_S, \D_B)$ as the set of refinements $M\in \Delta(\Delta(\Theta_B))$ with $\mu(M) = \D_B$ such that 
$$\E_{\D'_B \sim M} [\pi_i^{t-1}(\D_S, \D'_B)] \geq \pi_i^{t-1}(\D_S, \D_B), \qquad  \forall i \in \{B,S\} $$
Similarly in even rounds we define $\VR_S^t(\D_S, \D_B)$ as the set of refinements $M\in \Delta(\Delta(\Theta_S))$ with $\mu(M) = \D_S$ such that 
$$\E_{\D'_S \sim M} [\pi_i^{t-1}(\D'_S, \D_B)] \geq \pi_i^{t-1}(\D_S, \D_B), \qquad  \forall i \in \{B,S\} $$

\paragraph{Sub-protocol optimality} The second condition for equilibrium is that each message maximizes the payoff of the sender among all messages satisfying voluntary communication:

\begin{equation}\label{eq:equilibrium}
\begin{aligned}
&
\pi_B^t(\D_S, \D_B) = \max_{M\in \VR_B^t(\D_S, \D_B)} \E_{\D'_B \sim M} [\pi_B^{t-1}(\D_S, \D'_B)] \quad & &  \text{for odd } t \\
& \pi_S^t(\D_S, \D_B) = \max_{M\in \VR_S^t(\D_S, \D_B)} \E_{\D'_S \sim M} [\pi_S^{t-1}(\D'_S, \D_B)] \quad  & & \text{for even } t \\
\end{aligned}
\end{equation}

\paragraph{Equilibrium selection} The equilibrium conditions \eqref{eq:voluntary_refinement} and \eqref{eq:equilibrium} don't specify an unique protocol, since there may be multiple optimal information refinements for Sally leading to different utilities for Bob and vice-versa. For the remainder of the paper we will study the equilibrium in which each agent breaks ties in favor of the other one.\footnote{ Formally, this means that in odd rounds $t$ Bob will choose a refinement $M \in \VR_B^t(\D_S, \D_B)$ that lexicographically maximizes $(\E_{\D'_B \sim M} [\pi_B^{t-1}(\D_S, \D'_B)], \E_{\D'_B \sim M} [\pi_S^{t-1}(\D_S, \D'_B)])$ and Sally will choose a refinement $M \in \VR_S^t(\D_S, \D_B)$ that lexicographically maximizes $(\E_{\D'_S \sim M} [\pi_S^{t-1}(\D'_S, \D_B)], \E_{\D'_S \sim M} [\pi_B^{t-1}(\D'_S, \D_B)])$.}

With this tie-breaking in place, the values of $\pi^t_i(\D_S, \D_B)$ in equilibrium are uniquely determined. From now on, whenever we refer to  $\pi^t_i(\D_S, \D_B)$ we will be referring to the values in equilibrium with that tie-breaking rule.

\paragraph{Note on Voluntary Communication} In settings such as binary type spaces ($\abs{\Theta_S} = \abs{\Theta_B} = 2$) any refinement is voluntary. This is due to the fact that uncertainty can be represented by a single-parameter. We show this in Appendix \ref{sec:voluntary_discussion} together with a discussion on Voluntary Communication.

\subsubsection{Generic protocols}

One could consider more general protocols, for example by allowing both agents to simultaneously send messages and by considering messages in a generic space instead of just information refinements. Later in Appendix \ref{sec:reductions} we will show any general protocol can be reduced to an alternating information refinement protocol. We give the intuition below.

Information refinement is enough since the only use of messages is for agents to perform a Bayesian update on their information about the other agent. Hence it is enough to reason about the information updates directly.

To see that alternating communication is without loss of generality, observe that only Sally has an action at $t=0$. Hence there is no need for Sally to communicate at $t=1$ since Bob has no subsequent action to take. Sally may very well stay silent at time $t=1$. If that happens there is no need for Bob to say anything at $t=2$ since he can't influence any further message from Sally. Instead Bob can combine his communication at $t=2$ and $t=1$ and send it together at time $t=1$ staying silent at $t=2$. The same argument can be applied recursively showing that we can obtain a protocol with the same effect by having Bob and Sally speaking in alternating rounds.

\subsection{Message Complexity}\label{sec:message_complexity}

By voluntary communication, the utilities of the agents are monotone along the protocol, i.e., $\pi_i^t(\D_S, \D_B) \geq \pi_i^{t-1}(\D_S, \D_B)$ for all $i \in \{B,S\}$ and $t > 0$. This allows us to define \emph{limit utilities} as follows:
$$\pi_i^\infty(\D_S, \D_B) = \lim_{t \to \infty} \pi^t_i(\D_S, \D_B).$$
With that definition, we can define the \emph{message complexity} of a pair of distributions $(\D_S, \D_B)$ as:
$$\MC(\D_S, \D_B) := 
\min \left\{t \in \mathbb{N} \cup \{\infty\} \ \middle| \ \pi^t_S(\D_S, \D_B) + \pi^t_B(\D_S, \D_B) = \pi^\infty_S(\D_S, \D_B) + \pi^\infty_B(\D_S, \D_B)\right\}.$$
which corresponds to the minimum number of messages to achieve the best possible utilities.

Finally, if in equilibrium there is a message that moves from state $(\D_S, \D_B)$ in period $t$ to $(\D'_S, \D'_B)$ in period $t-1$ we say that $(\D'_S, \D'_B)$ is a \emph{child} of $(\D_S, \D_B)$ in period $t$.

\section{Efficient Bilateral Trade} \label{sec:bilateral_trade}

Bilateral trade will provide us with a concrete setting in which multiple rounds of interaction can lead to a better outcome -- in particular the efficient allocation. We start with a numerical example.

\subsection{Numerical example with one seller type}\label{sec:one_sided_uncertainty}

Recall that we are in the setting described in Section \ref{ex:bilateral_trade}. Consider the setting where Sally has a single type $\theta_s = 2$ and Bob has two possible types $3$ and $6$ with probability $p=1/3$ of having the high type. In the absence of any communication Sally will set a the price equal to $6$ and sell with probability $1/3$, resulting in inefficient trade: 
$$\pi^0_S = \frac{4}{3} \qquad \pi^0_B = 0 \qquad W^0 = \frac{4}{3} < W^* = 2$$
If we allow one round of communication, Bob can send a signal refining Sally's information about his type.\footnote{As we discussed in Section \ref{subsec:refinement} information refinements can always be interpreted as a Bayesian update of a signal correlated with the sender's type. In this particular case, Bob can accomplish this information refinement by sending one of two signals $\{L, H\}$. If Bob has the high type he can send the signal $H$ with probability $1/3$ and $L$ with probability $2/3$, and whenever he has the low type he always sends $L$. One can readily verify that this leads to the refinement described.} With probability $8/9$ Bob will send a message that refines Sally's information to $p=1/4$, where $p$ is the probability of having the high type. With probability $1/9$ Bob can refine Sally's information to $p=1$. In the first case, Sally is indifferent between the two prices and may very well price at $3$. In the second case, Sally will price at $6$. In either case, she will sell with probability one, leading to efficient trade. The outcome is the following:
$$\pi^1_S = \frac{4}{3} \qquad \pi^1_B = \frac{2}{3} \qquad W^1 = 2 = W^* = 2$$
This is in fact a general phenomenon described in Bergemann, Brooks and Morris \cite{bergemann2015limits}, who show that if there is no uncertainty about the seller's type, the buyer can always signal in order to extract the full surplus of the trade. The buyer does so by refining his distribution into revenue-equivalent distributions, i.e., distributions where the seller is indifferent between pricing at any point in the support. Below we state their result in our language, which will prove useful later.

\begin{lemma}[Bergemann-Brooks-Morris \cite{bergemann2015limits}]\label{lemma:bbm}
If $\D_S$ has a single point in the support, then $$\pi_B^1(\D_S, \D_B) = W^*(\D_S, \D_B) - \pi_B^0(\D_S, \D_B) \qquad \pi_S^1(\D_S, \D_B) =  \pi_S^0(\D_S, \D_B).$$
\end{lemma}

\subsection{Numerical example with two seller types}

If Bob also has uncertainty about the Sally's type, then it is no longer possible to achieve efficiency by one round of signaling. Let's keep Bob's types as $\Theta_B = \{3,6\}$ with $1/3$ probability on the high type. But now Sally's type 
is in $\Theta_S = \{0,2\}$ with $1/2$ probability on the high type.

Without any communication, Sally sets price $6$ whenever her type is $2$ and $3$ whenever her cost is $0$, leading to the following outcome:
$$\pi^0_S = \frac{13}{6} \qquad \pi^0_B = \frac{1}{2} \qquad W^0 = \frac{8}{3}  < W^* = 3$$

To describe the outcome with communication it is useful to use the diagram in Figure \ref{fig:two_types_example}. Since we have two types for each agent, we can represent the distributions $(\D_B, \D_S)$ by a pair $(p,q) \in [0,1]$ where $p$ is the probability that Bob has the high type and $q$ is the probability that Sally has the high type. Hence the initial distribution corresponds to the point $(1/3,1/2)$.

\begin{figure}[h]
\centering
\begin{tikzpicture}[scale=4]
    \draw[dashed] (0,.5)--(1,.5);
    \draw[dashed] (1/3,0)--(1/3,1);
    \draw[dashed] (1/2,0)--(1/2,1);
    \draw[dashed] (1/4,0)--(1/4,1);
    \node at (-.07,.5) {$\frac{1}{2}$};
    \node at (1/3,-.08) {$\frac{1}{3}$};
    \node at (1/4,-.08) {$\frac{1}{4}$};
    \node at (1/2,-.08) {$\frac{1}{2}$};
    \node at (-.07,.95) {$q$};
    \node at (.95,-.08) {$p$}; 
  \draw (0,0) -- (1,0) -- (1,1) -- (0,1) -- cycle;
    \draw [-{Latex[length=1.5mm,width=2mm]}, line width=1pt, color=blue] (1/3,1/2) -- (1/4,1/2);
    \draw [-{Latex[length=1.5mm,width=2mm]}, line width=1pt, color=blue] (1/3,1/2) -- (1/2,1/2);
   \node[circle,fill,inner sep=1pt] at (1/3,1/2) {};
   \node[circle,fill,inner sep=1pt] at (1/4,1/2) {};
   \node[circle,fill,inner sep=1pt] at (1/2,1/2) {};

 \begin{scope}[xshift=1.5cm]
    \draw[dashed] (0,.5)--(1,.5);
    \draw[dashed] (0,2/3)--(1,2/3);
    \draw[dashed] (1/3,0)--(1/3,1);
    \draw[dashed] (1/2,0)--(1/2,1);
    \draw[dashed] (1/4,0)--(1/4,1);
    \node at (-.07,.5) {$\frac{1}{2}$};
    \node at (-.07,2/3) {$\frac{2}{3}$};
    \node at (1/3,-.08) {$\frac{1}{3}$};
    \node at (1/4,-.08) {$\frac{1}{4}$};
    \node at (1/2,-.08) {$\frac{1}{2}$};
    \node at (-.07,.95) {$q$};
    \node at (.95,-.08) {$p$}; 
  \draw (0,0) -- (1,0) -- (1,1) -- (0,1) -- cycle;
  \draw [-{Latex[length=1.5mm,width=2mm]}, line width=1pt, color=red] (1/3,1/2) -- (1/3,0);
  \draw [-{Latex[length=1.5mm,width=2mm]}, line width=1pt, color=red] (1/3,1/2) -- (1/3,2/3);
  \draw [-{Latex[length=1.5mm,width=2mm]}, line width=1pt, color=blue] (1/3,2/3) -- (1/4,2/3);
  \draw [-{Latex[length=1.5mm,width=2mm]}, line width=1pt, color=blue] (1/3,2/3) -- (1,2/3);
   \node[circle,fill,inner sep=1pt] at (1/3,1/2) {};
   \node[circle,fill,inner sep=1pt] at (1/3,2/3) {};
   \node[circle,fill,inner sep=1pt] at (1/3,0) {};
   \node[circle,fill,inner sep=1pt] at (1/4,2/3) {};
   \node[circle,fill,inner sep=1pt] at (1,2/3) {};
 \end{scope}
\end{tikzpicture}
\caption{Each square corresponds to $[0,1]^2$. A point $(p,q)$ depicts a state where Bob has the high type with probability $p$ and Sally has the high type with probability $q$. The initial state is $(\frac{1}{3},\frac{1}{2})$. The arrows correspond to information refinements in the optimal $1$-round (left) and $2$-round (right) protocols. Red arrows correspond to Sally's refinement of Bob's information and blue arrows to Bob's refinement of Sally's information.}
\label{fig:two_types_example}
\end{figure}
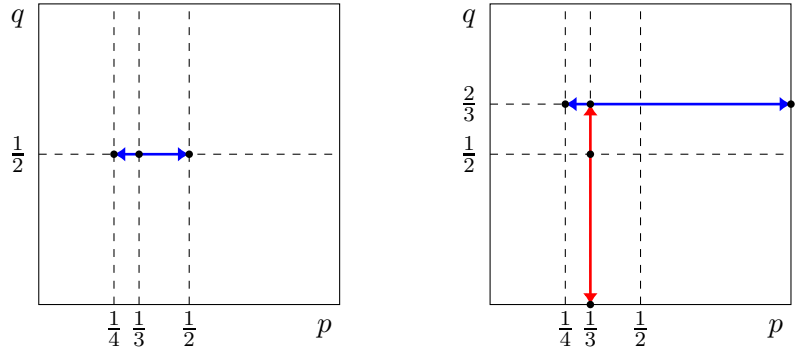

With one round of communication, the best thing that Bob can do is to refine\footnote{The reader is invited to check that this is a valid refinement (as defined in Section \ref{subsec:refinement}) since $\frac{1}{3}\cdot\frac{1}{2} + \frac{2}{3}\cdot\frac{1}{4} = \frac{1}{3}$.} Sally's information $p=1/2$ with probability $1/3$ and $p=1/4$ with probability $2/3$ (see the diagram on the left in Figure \ref{fig:two_types_example}). We will explain in the next subsection why this is the optimal choice for Bob. Now Sally's price depends both on her type and the information she has about Bob. Evaluating the four cases we get to the following outcome:
$$\pi^1_S = \frac{13}{6} \qquad \pi^1_B = \frac{3}{4} \qquad W^1 = \frac{35}{12} < W^* = 3$$

Bob's signal improves the efficiency of the allocation and extracts the additional efficiency as buyer surplus, but it is not quite enough to achieve full efficiency. With two rounds of communication, however, we obtain the efficient allocation. In the diagram on the right in Figure \ref{fig:two_types_example} we depict the optimal two-round protocol: first Sally sends a signal that refines Bob's signal to $q=0$ with probability $1/4$ and $q=2/3$ with probability $3/4$. If Bob's receives the $q=0$ signal, he stays silent since he knows Sally will already price at the low type. If Bob receives the $q=2/3$ signal, however, he refines Sally's information to $p=1/4$ with probability $8/9$ and $p=1$ with probability $1/9$. This leads to the following outcome:
$$\pi^2_S = \frac{9}{4} \qquad \pi^2_B = \frac{3}{4} \qquad W^2 = 3 =  W^* = 3$$
By revealing some information about her type, Sally incentivizes Bob to reveal more about his, leading to a distribution of information points in which efficient trade is possible. In the next section we show that this is a general phenomenon.


\subsection{Nonstop Communication until Efficient Trade}\label{subsec:main_result}

Our main result is that agents don't stop talking until they reach an efficient outcome. Formally, we will show that if the agents have no further use for rounds of communication after a certain round (finite message complexity) then it must be because they have reached an efficient outcome.

\begin{theorem}[Main Theorem]\label{thm:main_trade_thm}
If the message complexity $\MC(\D_S, \D_B) = t < \infty$ for bilateral trade, then the equilibrium allocation is efficient, i.e., $$\pi^t_S(\D_S, \D_B) + \pi^t_B(\D_S, \D_B) = W^*(\D_S, \D_B)$$
where $W^*(\D_S, \D_B)$ is the welfare of the optimal allocation.
\end{theorem}


We will apply a recursive argument. First we show in Lemma \ref{lem:zero_complexity} that if the allocation can't improve with any communication (zero message complexity) then the allocation must already be efficient. In the recursive step, we argue that if there are distributions $(\D_S, \D_B)$ which both have finite message complexity and lead to an inefficient outcome, then it is possible to construct distributions $(\D'_S, \D'_B)$ with strictly smaller message complexity that also lead to an inefficient outcome.



\begin{lemma}\label{lem:zero_complexity}
If $\MC(\D_S, \D_B) = 0$, then the equilibrium allocation is efficient, i.e., $$\pi^0_S(\D_S, \D_B) + \pi^0_B(\D_S, \D_B) = W^*(\D_S, \D_B).$$
\end{lemma}
\begin{proof}
Let $F(\D_S)$ be the message where Sally reveals her true cost (full revelation). Then we know that:
$$ \pi_S^1(\D_S,\D_B) = \pi_S^0(\D_S,\D_B) = \E_{\D'_S \sim F(\D_S)} [\pi_S^0(\D'_S,\D_B)] = \E_{\D'_S \sim F(\D_S)} [\pi_S^1(\D'_S,\D_B)]$$
where the first equality comes from the fact that $\MC(\D_S, \D_B) = 0$, the second equality follows from the fact that $\pi_S^0$ is linear in $\D_S$ and the third follows from Lemma \ref{lemma:bbm}, since $\D'_S$ has only a single point in its support. Hence Sally has a message in period $t=2$ that keeps her utility unchanged but improves Bob's utility if the allocation is inefficient. Therefore, the message complexity cannot be $0$.
\end{proof}

Finally, we describe how the message complexity of a node in the protocol relates to the message complexity of its children. Recall the definition of a \emph{child} from Section \ref{sec:message_complexity}.

\begin{lemma}
If $\MC(\D_S, \D_B) = t$ with $1 \leq t < \infty$, then any child of $(\D_S, \D_B)$ in period $t$ must have message complexity at most $t - 1$.
\label{lem:complexity_plus_one}
\end{lemma}
\begin{proof}
We prove by contradiction. Assume $t$ is even (Sally's turn to communicate) and $(\D'_S, \D_B)$ is a child of $(\D_S, \D_B)$ with message complexity $t' \geq t$. Then there is some finite $T$ with the same parity as $t - 1$ (to make it the same player's turn) satisfying $t \leq T \leq t' + 1$, 
$$\pi^{T}_S(\D'_S, \D_B) + \pi^{T}_B(\D'_S, \D_B) > \pi^{t-1}_S(\D'_S, \D_B) + \pi^{t-1}_B(\D'_S, \D_B)$$ 
by the definition of message complexity. Furthermore, by voluntary communication we have that:
$$\pi^{T}_S(\D'_S, \D_B) \geq \pi^{t-1}_S(\D'_S, \D_B), \qquad \pi^{T}_B(\D'_S, \D_B) \geq \pi^{t-1}_B(\D'_S, \D_B).$$

Therefore, either $\pi^{T}_S(\D'_S, \D_B) > \pi^{t-1}_S(\D'_S, \D_B)$, or $\pi^{T}_S(\D'_S, \D_B) = \pi^{t-1}_S(\D'_S, \D_B)$ and $\pi^{T}_B(\D'_S, \D_B) > \pi^{t-1}_B(\D'_S, \D_B)$. In the first Sally could use her message of the $t$-th period in the $(T+1)$-th period, obtaining:
$$\pi^{T+1}_S(\D_S, \D_B) \geq \E_{\D'_S \sim m_S^{t}(\D_S, \D_B)} [\pi^{T}_S(\D'_S, \D_B)] > \E_{\D'_S \sim m_S^{t}(\D_S, \D_B)} [\pi^{t-1}_S(\D'_S, \D_B)] = \pi^{t}_S(\D_S, \D_B).$$
Or in the second case, she can use the same construction to obtain an outcome that gives her the same utility while benefiting Bob, leading to $\pi^{T+1}_S(\D_S, \D_B) + \pi^{T+1}_B(\D_S, \D_B) > \pi^{t}_S(\D_S, \D_B) + \pi^{t}_B(\D_S, \D_B)$.\footnote{If Sally has a strategy to strictly improve her utility, we have $\pi^{T+1}_S(\D_S, \D_B) > \pi^t_S(\D_S, \D_B)$. Otherwise, Bob's utility must satisfy $\pi^{T+1}_B(\D_S, \D_B) > \pi^{t}_B(\D_S, \D_B)$. There is always $\pi^{T+1}_S(\D_S, \D_B) \geq \pi^t_S(\D_S, \D_B)$ and $\pi^{T+1}_B(\D_S, \D_B) \geq \pi^t_B(\D_S, \D_B)$ by voluntary communication.} In either case we have a contradiction with the premise that $\MC(\D_S, \D_B) = t$.

The same argument holds if $t$ is odd and $(\D_S, \D'_B)$ is a child with message complexity $t' \geq t$ by replaying the argument reversing the roles of Sally and Bob.
\end{proof}

\begin{proof}[Proof of Theorem \ref{thm:main_trade_thm}]
Define the following sets of pairs of distributions:
$$\Gamma = \{ (\D_S, \D_B) \mid \MC(\D_S, \D_B) < \infty \text{ and } \pi^\infty_S(\D_S, \D_B) + \pi^\infty_B(\D_S, \D_B) < W^*(\D_S, \D_B) \}.$$
We want to show that $\Gamma = \emptyset$. If $\Gamma$ is not empty, choose any pair $(\D_S, \D_B)$ with smallest possible message complexity $t=\MC(\D_S, \D_B)$. We know by Lemma \ref{lem:zero_complexity} that $t \geq 1$ hence state $(\D_S, \D_B)$ must have an inefficient child $(\D'_S, \D_B)$ or $(\D_S, \D'_B)$ in period $t-1$. By Lemma~\ref{lem:complexity_plus_one} this pair must have message complexity at most $t-1$ contradicting the minimality of $(\D_S, \D_B)$.
\end{proof}

\subsection{Efficient Communication Conjecture}

\newtheorem*{ecc}{Efficient Communication Conjecture}

Theorem \ref{thm:main_trade_thm} leaves it open the possibility that there are instances of bilateral trade where the allocation indefinitely improves with the increasing rounds of the communication protocol. This is an intriguing possibility, but based on evidence from computer simulations, we conjecture that an efficient allocation is always reached in finitely-many rounds of communication. We name it the ``Efficient Communication Conjecture'' (ECC) and formalize it below:

\begin{ecc}
In the bilateral trade setting, if $\Theta_B$ and $\Theta_S$ are finite types spaces, then for any distributions $\D_B \in \Delta(\Theta_B)$ and $\D_S \in \Delta(\Theta_S)$ the message complexity is finite: $\MC(\D_B, \D_S) < \infty$.
\end{ecc}

Resolving in either way would lead to an interesting result. Either we show that efficient allocation is always reached after finitely-many rounds or we exhibit the possibility where agents would prefer to have an infinitely-long conversation before trade. 

In the Section \ref{sec:binary_types} we make progress on this question by showing an important class of bilateral trade instances where an efficient allocation is achieved with two rounds of communication. Whenever Bob's type space is binary ($\abs{\Theta_B}=2$) and Sally's type space is finite ($\abs{\Theta_S}< \infty$), we show that for any distributions $\D_B$ and $\D_S$ supported on $\Theta_B$ and $\Theta_S$ respectively, we have $\MC(\D_B, \D_S) \leq 2$. In that case, we can completely describe the communication protocol.

Whenever Bob's type space is larger ($\abs{\Theta_B}>2$), two messages are no longer enough. Consider the following example where $\Theta_B = \{3,6,12\}$ and $\Theta_S = \{0,2\}$. Let $\D_B$ be the uniform distribution over $\Theta_B$ and $\Theta_S$ be the distribution that puts $1/5$ probability on the high type. The welfare of the efficient allocation is $W^* = 33/5$. In the table below we write the payoffs of Sally and Bob and the welfare for $t=0,1,2$.

\def\arraystretch{1.5}
\begin{center}
\begin{tabular}{ c|c|c|c } 
  & $t=0$ & $t=1$ & $t=2$ \\ 
  \hline
  S & 58/15 & 58/15 & 62/15  \\ 
  B & 8/5 & 12/5 & 12/5  \\ 
  W & 82/15 & 94/15 & 98/15   \\ 
\end{tabular}
\end{center}

Observe that the welfare is strictly smaller than $W^*$. By Theorem \ref{thm:main_trade_thm} we have that further rounds of communication must improve the welfare and therefore $\MC(\D_S, \D_B) \geq 3$. In Section \ref{sec:three_rounds} of the appendix we present the details of this example.

\subsection{Short protocol for binary buyer types}\label{sec:binary_types}

For an important class of bilateral trade instances we show that a short (two-round) communication protocol leads to efficiency. For this class, we will be able to fully characterize the communication protocol. The assumption will be that the buyer has a binary type space, i.e. $\abs{\Theta_B}=2$. In order to construct the communication protocol, we will develop a higher-order indifference argument. The original indifference argument in \cite{bergemann2015limits} decomposes a prior into posterior distributions where the seller is indifferent between different prices. Here: since the seller knows that the buyer will signal in order to make her indifferent, she will send a preliminary signal that makes the buyer indifferent among different ways to make the seller indifferent.

\begin{theorem}[Efficient Communication for Buyer Binary Types] \label{thm:buyer_binary}
Let $\Theta_B = \{v_1, v_0\}$ and $\Theta_S = \{c_1, c_2, \hdots, c_n\}$. Then in two rounds of communication, Bob and Sally achieve an efficient allocation.
\end{theorem}

The rest of this section is dedicated to proving Theorem \ref{thm:buyer_binary}. We will use $p$ to denote the probability that the buyer has the high type and let $q$ be a vector in the $n$-simplex $\Delta_n = \{q \in [0,1]^n; \sum_i q_i = 1\}$ representing the probability that Sally has each type. This way, we can represent the pair $(\D_B, \D_S)$ by $(p,q)$.

Without loss of generality we will also assume that:
$$c_1 < c_2 < \cdots < c_n < v_0 < v_1$$
since any of Sally's type such that $c_i > v_0$ is ignored by Bob since she will price at $v_1$ or higher anyway if she has that type, regardless of what her information about Bob's type is. 

To prove this theorem, we will first characterize $\pi^0(p,q)$ and then obtain $\pi^1(p,q)$ and $\pi^2(p,q)$ using the equilibrium definition (equation \eqref{eq:equilibrium}). We will view this as a process of alternating concavification: we alternatively replace the payoffs of Bob and Sally by their concave hulls.

\subsubsection{Characterizing $\pi^0$}
Without any communication, Sally's decision only depends on $p$ and her type, and hence the payoffs will be linear in $q$, i.e.:
$$\pi_i^0(p,q) = \sum_{j=1}^n q_j\cdot \pi_i^0(p,e_j), \quad \forall (p,q) \in [0,1] \times \Delta_n, i \in \{B,S\}$$
where $e_j$ is the $j$-th unit vector. Therefore, we can focus on understanding the cases for $q = e_j$. If Sally has cost $c_j$ and Bob has probability $p$ of being the high type, Sally will price at $v_1$ whenever:
$$p (v_1 - c_j) > v_0 - c_j$$
and at $v_0$ otherwise. This gives the threshold $p^*_j$ defined as follows:
$$p^*_j = \frac{v_0 - c_j}{v_1 - c_j}$$
It will be convenient in the following analysis to set $p^*_0 = 1$. 

We can plot the payoff curves for Bob and Sally at $q = e_j$ in Figure \ref{fig:pi0}. Before $p_j^*$, Sally prices at Bob's low type and get constant revenue equal to $v_0-c_j$. After $p_j^*$, Sally prices at Bob's high type and his payoff is zero. For each $q \in \Delta_n$, Bob's utility is a combination of $n$ curves like the ones in Figure \ref{fig:pi0} with breakpoints at $p_i^*$. See the left side of Figure \ref{fig:pi0_buyer}.

\begin{figure}[h]
\centering
\begin{tikzpicture}[scale=3.6] 
\begin{scope}[xshift=0cm]
    \draw (0,.9) -- (0,0) -- (1,0);
    \tikzmath{
        \q=0;
        \s=.15;
        \p1 = (\q + (1-\q)*3)*\s;
        \p2 = (\q + (1-\q)*3)*\s;
        \p3 = (2*\q + (1-\q)*3)*\s;
        \p4 = (4*\q + (1-\q)*6)*\s;
        {\draw [line width=1.5pt, color=red] (0,\p1) -- (1/4,\p2) -- (1/2, \p3) -- (1,\p4); };
        \s = 1/4;
        \p2 = (3/4)*\s;
        \p4 = ((1-\q)*(3/4))*\s;
        \p3 = ((1-\q)*(3/2))*\s;
        {\draw [line width=1.5pt, color=blue] (0,0) -- (1/4,\p2) -- (1/4, \p4) -- (1/2, \p3) -- (1/2, 0) -- (1,0); };
    };
    \node[circle,fill,inner sep=1pt] at (1,0) {};
    \node[circle,fill,inner sep=1pt] at (0,0) {};
    \node[circle,fill,inner sep=1pt] at (1/2,0) {};
    \node at (0,-.08) {$0$};
    \node at (1,-.08) {$1$};
    \node at (1/2,-.08) {$p^*_j$};
    \node at (.22,.85) {$\pi^0_i(p,e_j)$};
    \node at (.95,.07) {$B$};
    \node at (.95,.75) {$S$};
\end{scope}

\end{tikzpicture}
\caption{Payoffs of Bob (blue) and Sally (red) without communication for $q =e_j$.}
\label{fig:pi0}
\end{figure}
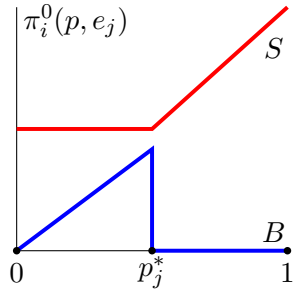

\begin{figure}[h]
\centering
\begin{tikzpicture}[scale=3.05] 
    \draw (0,.85) -- (0,0) -- (2,0);
    \node[circle,fill,inner sep=1pt] at (0,0) {};
    \node[circle,fill,inner sep=1pt] at (.5,0) {};
    \node[circle,fill,inner sep=1pt] at (.8,0) {};
    \node[circle,fill,inner sep=1pt] at (1.1,0) {};
    node[circle,fill,inner sep=1pt] at (1.7,0) {};
    \node[circle,fill,inner sep=1pt] at (2,0) {};
    \node at (0,-.08) {$0$};
    \node at (.5,-.08) {$p_4^*$};
    \node at (.8,-.08) {$p_3^*$};
    \node at (1.1,-.08) {$p_2^*$};
    \node at (1.7,-.08) {$p_1^*$};
    \node at (2,-.08) {$p_0^*=1$};

    \draw [line width=1.5pt, color=blue] (0,0) --
    (.5,1*.5)--(.5,.8*.5)--
    (.8,.8*.8)--(.8,.8*.4)--
    (1.1,1.1*.4)--(1.1,1.1*.3)--
    (1.7,1.7*.3)--(1.7,0)--(2,0);
      \node at (.23,.85) {$\pi^0_B(p,q)$};

     \begin{scope}[xshift=2.3cm]
    \draw (0,.85) -- (0,0) -- (2,0);
    \node[circle,fill,inner sep=1pt] at (0,0) {};
    \node[circle,fill,inner sep=1pt] at (.5,0) {};
    \node[circle,fill,inner sep=1pt] at (.8,0) {};
    \node[circle,fill,inner sep=1pt] at (1.1,0) {};
    node[circle,fill,inner sep=1pt] at (1.7,0) {};
    \node[circle,fill,inner sep=1pt] at (2,0) {};
    \node at (0,-.08) {$0$};
    \node at (.5,-.08) {$p_4^*$};
    \node at (.8,-.08) {$p_3^*$};
    \node at (1.1,-.08) {$p_2^*$};
    \node at (1.7,-.08) {$p_1^*$};
    \node at (2,-.08) {$p_0^*=1$};

    \draw [line width=1.5pt, color=blue, dashed] (0,0) --
    (.5,1*.5)--(.5,.8*.5)--
    (.8,.8*.8)--(.8,.8*.4)--
    (1.1,1.1*.4)--(1.1,1.1*.3)--
    (1.7,1.7*.3)--(1.7,0)--(2,0);
    \draw [line width=1.5pt, color=blue] (0,0) --
    (.5,1*.5)--(.8,.8*.8)--
    (1.7,1.7*.3)--(2,0);
    \node at (.23,.85) {$\pi^1_B(p,q)$};
    \end{scope}

\end{tikzpicture}
\caption{Bob's best response in $t=1$ for a generic distribution $q$.}
\label{fig:pi0_buyer}
\end{figure}
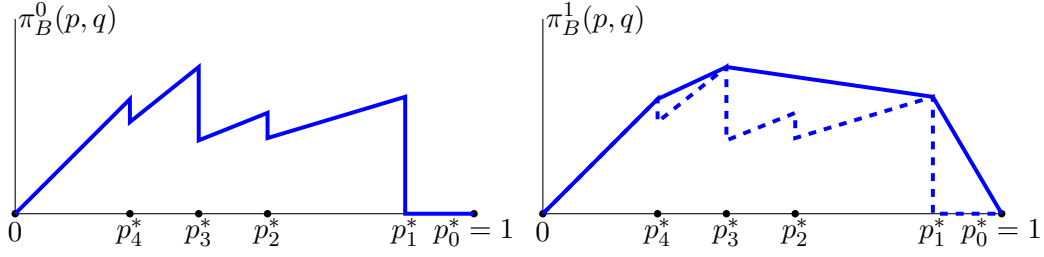

\subsubsection{Characterizing $\pi^1$}
\label{sec:pi_1}
From the previous subsection we understand the outcome if we reach the last round ($t=0$) for each pair $(p,q)$. From that we can infer what should be Bob's response in the second-to-last round ($t=1$) for each pair $(p,q)$. By an information refinement, Bob can obtain a convex combination of payoffs $\pi^0_B(p',q)$ as long as the distributions average to $p$. In other words, we can translate the equilibrium condition \eqref{eq:equilibrium} as:
$$\pi^1_B(p,q) = \max_\lambda \sum_u \lambda_u \cdot \pi^0_B(p_u,q) \quad \text{s.t} \quad \sum_u \lambda_u \cdot p_u = p, \quad \sum_u \lambda_u = 1, \quad \lambda_u \in [0,1]$$
which means that $\pi^1_B(p,q)$ is the concave hull of $\pi^0_B(p,q)$ along the $p$-coordinate. The convex combination coefficients $\lambda_u$ tell us how to obtain $\pi^1_S(p,q)$ from $\pi^0_S(p,q)$.

For any fixed $q$, the function $p \mapsto \pi^0_B(p,q)$ is piecewise linear with peaks at $p_j^*$. Bob's signal at $t=1$ is obtained by taking the concave hull of his utility curve and signaling using the endpoints (right side of Figure \ref{fig:pi0_buyer}). For the example in the figure, Bob will use the following strategy:
\begin{itemize}
    \item Stay silent if $p \leq p^*_4$.
    \item If $p \in [p^*_4, p^*_3]$, refine $p$ to the endpoints $p^*_4$ and $p^*_3$.
    \item If $p \in [p^*_3, p^*_1]$, refine $p$ to the endpoints $p^*_3$ and $p^*_1$.
    \item If $p \in [p^*_1, 1]$, refine $p$ to the endpoints $p^*_1$ and $1$.
\end{itemize}


\subsubsection{Characterizing $\pi^2$}
We will show that at $t=2$, Sally has a signal that keeps Bob's payoffs unchanged and captures the remaining surplus from the optimal allocation:
\begin{equation}\label{eq:buyer_binary_sally2}
\pi^2_S(p,q) = W^*(p,q) - \pi^1_S(p,q) \qquad \pi^2_B(p,q) =  \pi^1_B(p,q)
\end{equation}

To construct this signal, we will first define one  \emph{indifference distribution} for each possible support of Sally's distribution. For each non-empty subset $X = \{x(1), x(2), \hdots, x(k) \} \subseteq [n]$ define the distribution $q^X$ over $[n]$ as follows (using $p_0^* = 1$ and $x(0) = 0$):
$$q^X_{x(i)} = \frac{ \big(p^*_{x(i)}\big)^{-1} - \big(p^*_{x(i-1)}\big)^{-1} }{ \big(p^*_{x(k)}\big)^{-1} - 1} \text{ for } i=1, \ldots, k \quad \text{ and } \quad
q_j^X = 0 \text{ for } j \notin X$$
The advantage of indifference distributions is that the peaks of $\pi_B^0(p,q^X)$ are aligned (see Figure \ref{fig:pi0_buyer_indiff}) and hence if Bob's information about Sally is $q^X$ at $t=1$, then Bob's optimal strategy is the following:
\begin{itemize}
    \item If $p \leq p_{x(k)}^*$, Bob will stay silent.
    \item If $p > p_{x(k)}^*$, Bob will refine his signal to either $1$ or $p_{x(k)}^*$.
\end{itemize}
In either case, we are at an efficient allocation. If $p \leq p_{x(k)}^*$, then Sally prices at the lowest signal of Bob's support. If $p=1$ then Bob has the high type with probability $1$ and Sally prices at that point. Hence:
\begin{equation}\label{eq:bob_response_indiff}
\pi_S^1(p,q^X) + \pi_B^1(p,q^X) = W^*(p, q^X)
\end{equation}

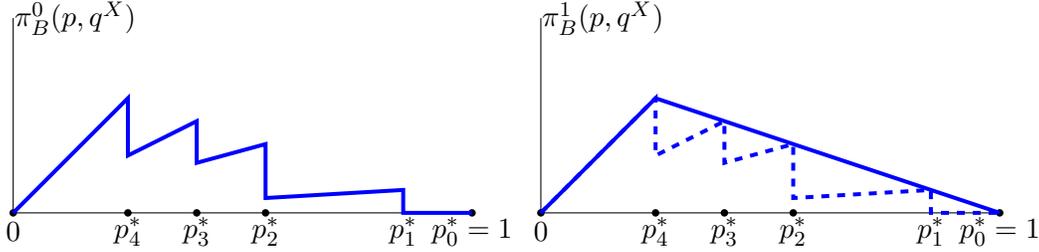
\begin{figure}[h]
\centering
\begin{tikzpicture}[scale=3.05] 
    \draw (0,.85) -- (0,0) -- (2,0);
    \node[circle,fill,inner sep=1pt] at (0,0) {};
    \node[circle,fill,inner sep=1pt] at (.5,0) {};
    \node[circle,fill,inner sep=1pt] at (.8,0) {};
    \node[circle,fill,inner sep=1pt] at (1.1,0) {};
    node[circle,fill,inner sep=1pt] at (1.7,0) {};
    \node[circle,fill,inner sep=1pt] at (2,0) {};
    \node at (0,-.08) {$0$};
    \node at (.5,-.08) {$p_4^*$};
    \node at (.8,-.08) {$p_3^*$};
    \node at (1.1,-.08) {$p_2^*$};
    \node at (1.7,-.08) {$p_1^*$};
    \node at (2,-.08) {$p_0^*=1$};
    

    \draw [line width=1.5pt, color=blue] (0,0) --
    (.5,1*.5)--(.5,.5*.5)--
    (.8,.5*.8)--(.8,.2727*.8)--
    (1.1,1.1*.2727)--(1.1,1.1*0.0588)--
    (1.7,1.7*0.0588)--(1.7,0)--(2,0);
     \node at (.27,.85) {$\pi^0_B(p,q^X)$};

     \begin{scope}[xshift=2.3cm]
    \draw (0,.85) -- (0,0) -- (2,0);
    \node[circle,fill,inner sep=1pt] at (0,0) {};
    \node[circle,fill,inner sep=1pt] at (.5,0) {};
    \node[circle,fill,inner sep=1pt] at (.8,0) {};
    \node[circle,fill,inner sep=1pt] at (1.1,0) {};
    node[circle,fill,inner sep=1pt] at (1.7,0) {};
    \node[circle,fill,inner sep=1pt] at (2,0) {};
    \node at (0,-.08) {$0$};
    \node at (.5,-.08) {$p_4^*$};
    \node at (.8,-.08) {$p_3^*$};
    \node at (1.1,-.08) {$p_2^*$};
    \node at (1.7,-.08) {$p_1^*$};
    \node at (2,-.08) {$p_0^*=1$};

    \draw [line width=1.5pt, color=blue, dashed] (0,0) --
    (.5,1*.5)--(.5,.5*.5)--
    (.8,.5*.8)--(.8,.2727*.8)--
    (1.1,1.1*.2727)--(1.1,1.1*0.0588)--
    (1.7,1.7*0.0588)--(1.7,0)--(2,0);
    \draw [line width=1.5pt, color=blue] (0,0) --
    (.5,1*.5)--(2,0);
    \node at (.27,.85) {$\pi^1_B(p,q^X)$};
    \end{scope}

\end{tikzpicture}
\caption{Bob's best response in $t=1$ for an indifference distribution $q$.}
\label{fig:pi0_buyer_indiff}
\end{figure}

With that we are ready to prove the statement of equation \eqref{eq:buyer_binary_sally2} which directly implies Theorem \ref{thm:buyer_binary}. We will decompose the proof in two lemmas:

\begin{lemma}\label{lemma:indifference_decomposition}
Let $q$ be any distribution over $[n]$. Then there exists a decomposition of $q$ into nested indifference distributions: there are subsets $$[n] \supseteq X_1 \supseteq X_2 \supseteq  \cdots \supseteq X_k \neq \emptyset$$ such that:
$$q = \sum_{i=1}^k \lambda_i q^{X_i} \quad \text{s.t.} \quad \sum_{i=1}^k \lambda_i = 1 \text{ and } \lambda_i \geq 0$$. 
\end{lemma}

\begin{proof}
Let $X$ be the support of $q$ and define $z = \min_{i \in X} q^X_i / q_i$. Since both $q$ and $q^X$ have support $X$, the components $q^X_i$ and $q_i$ are strictly positive, hence $z > 0$. Also, observe that $z = \min_i q^X_i / q_i \leq (\sum_i q^X_i)/ (\sum_i q_i) = 1$ since both are distributions. If $z=1$ then $q=q^X$ and we are done since $q$ has a trivial decomposition into indifference distributions. 

If $z < 1$, define $q' \in \Delta_n$ as follows:
$$q' = \frac{q-z q^X}{1-z}$$
To see that $q'$ is a distribution, observe first that $q_i - z q^X_i \geq 0$ by the definition of $z$ and that:
$$\sum_i q'_i = \frac{1}{1-z} \sum_i (q_i - z q^X_i) = \frac{1-z}{1-z} = 1$$
Finally, note that if $i$ is the index such that $z =  q^X_i / q_i$ we have $q'_i = 0$. Hence we showed how to decompose $q$ into a distribution of the form $q^X$ and a distribution with support strictly contained in $X$. Applying the construction recursively for $q'$ we obtain the decomposition in the statement.
\end{proof}

\begin{lemma}
Sally's signal that at $t=2$ refines her distribution into indifference distributions is her best response and satisfied equation \eqref{eq:buyer_binary_sally2}.
\end{lemma}

\begin{proof}
Let's assume without loss of generality that $q$ has full support and that $p > p^*_n$. If $p \leq p^*_n$, then $\pi_S^0(p,q) + \pi_B^0(p,q) = W^*(p,q)$ since Sally always prices at $v_0$ so equation  \eqref{eq:buyer_binary_sally2} holds trivially.

Consider the decomposition of Sally's original distribution $q$ into nested indifference distributions in Lemma \ref{lemma:indifference_decomposition} $q = \sum_{i} \lambda_i q^{X_i}$. Then:
$$\begin{aligned}
\pi^2_S(p,q) & \geq \sum_i \lambda_i \pi^1_S(p,q^{X_i}) =
\sum_i \lambda_i [W^*(p,q^{X_i}) - \pi^1_B(p,q^{X_i})] = 
W^*(p,q) - \sum_i \lambda_i \pi^1_B(p,q^{X_i})
\end{aligned}$$
where the first equality comes from equation \eqref{eq:bob_response_indiff} and the second due to the fact that $W^*(p,q)$ is linear in $q$.

To complete the proof, we need to show that  $\sum_i \lambda_i \pi^1_B(p,q^{X_i}) = \pi^1_B(p,q)$. One of the directions is easy: if $M$ is the optimal refinement for Bob in $t=1$ then:
$$\sum_i \lambda_i \pi^1_B(p,q^{X_i}) \geq \sum_i \lambda_i \E_{p' \sim M}[\pi^0_B(p,q^{X_i})] = \E_{p' \sim M}[\pi^0_B(p,q)] = \pi^1_B(p,q) $$
where the second inequality follows from the fact that $\pi_B^0$ is linear in $q$. This inequality in particular, verifies that Sally's signal satisfies voluntary communication.

Showing the other direction is the delicate aspect of the proof and will use the fact that the sets $X_i$ are nested. Start by choosing the smallest $i$ such that $$\min_{x \in X_i} p^*_x > p$$
If no such index exists, set $i=k+1$ and $p^*_a = 1$. Otherwise, we know
by the assumption that $p > p^*_n$ and hence $i > 1$. Let $a = \max\{x; x \in X_i\}$ and $b = \min\{x; x \in X_{i-1}\}$.
By the definition we have that $p^*_b \leq p \leq p^*_a$. The crucial observation is that the payoff functions $p \mapsto \pi_B^1(p, q^{X_i})$ are linear for $p \in [p^*_b, p^*_a]$. The reason is that the $\pi_B^1(p, q^{X_i})$ is continuous and piecewise linear with two pieces with a breakpoint at $\min_{x \in X_i} p^*_x$. Hence the interval $[p^*_b, p^*_a]$ is always within the same linear piece. If $p = z_a p_a^* + z_b p_b^*$ then:
$$\begin{aligned}
\sum_i \lambda_i \pi^1_B(p,q^{X_i}) = \sum_i \lambda_i [z_a \pi^0_B(p^*_a,q^{X_i}) + z_b \pi^0_B(p^*_b,q^{X_i}) ] 
 =  z_a \pi^0_B(p^*_a,q) + z_b \pi^0_B(p^*_b,q) \leq \pi^1_B(p,q)
\end{aligned}$$
The first equality follows from the observation about the linearity of $\pi_B^1(p, q^{X_i})$ in $[p^*_b, p^*_a]$, the second follows from the linearity of $\pi_B^0(p,q)$ on $q$ and the third is simply because the refining $p$ to $p^*_b$ and $p^*_a$ is a valid strategy for Bob.
\end{proof}

\newcommand{\SpyBestResponse}{
\begin{figure}[h]
\centering
\begin{tikzpicture}[scale=4]

 \begin{scope}[xshift=0cm, scale=.9]
    \draw (0.333333,0.000000)--(0.333333,1.000000);
    \draw (0.666667,0.000000)--(0.666667,1.000000);
    \draw (0,0) -- (1,0) -- (1,1) -- (0,1) -- cycle;
    \node at (2/3, -.08) {$\nicefrac{2}{3}$};
    \node at (1/3, -.08) {$\nicefrac{1}{3}$};
    \node at (1/2, 1/2) {E/E};
    \node at (1/6, 1/2) {C/E};
    \node at (5/6, 1/2) {E/C};
    \node at (1/2, -.25) {$t=0$};
 \end{scope}
 
  \begin{scope}[xshift=1.2cm, scale=.9]
\draw (0.333333,0.000000)--(0.333333,0.444444);
\draw (0.333333,0.444444)--(0.333333,0.555556);
\draw (0.666667,0.444444)--(0.666667,0.555556);
\draw (0.666667,0.555556)--(0.666667,1.000000);
\draw (0.333333,0.444444)--(0.666667,0.444444);
\draw (0.666667,0.444444)--(1.000000,0.444444);
\draw (0.000000,0.555556)--(0.333333,0.555556);
\draw (0.333333,0.555556)--(0.666667,0.555556);
\draw (0,0) -- (1,0) -- (1,1) -- (0,1) -- cycle;
\node at (0.333333, -.08) {$\nicefrac{1}{3}$};
\node at (0.666667, -.08) {$\nicefrac{2}{3}$};
\node at ( -.08,0.444444) {$\nicefrac{4}{9}$};
\node at ( -.08,0.555556) {$\nicefrac{5}{9}$};
\draw[\arw-\arw, line width=1pt, color=blue] (0,7/9)--(2/3,7/9);


\draw[\arw-\arw, line width=1pt, color=blue] (1/3,2/9)--(1,2/9);

\draw[\arw-\arw, line width=1pt, color=blue] (1/3,1/2)--(2/3,1/2);
    \node at (1/2, -.25) {$t=1$};
 \end{scope}

 \begin{scope}[xshift=2.4cm, scale=.9]

\draw (0.000000,0.000000)--(0.000000,0.444444);
\draw (0.000000,0.444444)--(0.000000,0.555556);
\draw (0.000000,0.555556)--(0.000000,1.000000);
\draw (0.333333,0.000000)--(0.333333,0.444444);
\draw (0.333333,0.444444)--(0.333333,0.555556);
\draw (0.481481,0.000000)--(0.481481,0.444444);
\draw (0.481481,0.444444)--(0.481481,0.555556);
\draw (0.518519,0.444444)--(0.518519,0.555556);
\draw (0.518519,0.555556)--(0.518519,1.000000);
\draw (0.666667,0.444444)--(0.666667,0.555556);
\draw (0.666667,0.555556)--(0.666667,1.000000);
\draw (1.000000,0.000000)--(1.000000,0.444444);
\draw (1.000000,0.444444)--(1.000000,0.555556);
\draw (1.000000,0.555556)--(1.000000,1.000000);
\draw (0.000000,0.444444)--(0.000000,0.444444);
\draw (0.481481,0.444444)--(0.518519,0.444444);
\draw (0.518519,0.444444)--(0.666667,0.444444);
\draw (0.666667,0.444444)--(1.000000,0.444444);
\draw (1.000000,0.444444)--(1.000000,0.444444);
\draw (0.000000,0.555556)--(0.000000,0.555556);
\draw (0.000000,0.555556)--(0.333333,0.555556);
\draw (0.333333,0.555556)--(0.481481,0.555556);
\draw (0.481481,0.555556)--(0.518519,0.555556);
\draw (1.000000,0.555556)--(1.000000,0.555556);
\draw (0,0) -- (1,0) -- (1,1) -- (0,1) -- cycle;

\draw[\arw-\arw, line width=1pt, color=red] (13.5/27,4/9)--(13.5/27,5/9);
\draw[\arw-\arw, line width=1pt, color=red] (1-0.407,4/9)--(1-0.407,1);
\draw[\arw-\arw, line width=1pt, color=red] (5/6,4/9)--(5/6,1);

\draw[\arw-\arw, line width=1pt, color=red]
(0.407,5/9)--(0.407,0);
\draw[\arw-\arw, line width=1pt, color=red]
(1/6,5/9)--(1/6,0);

    \node at (1/2, -.25) {$t=2$};
 \end{scope}

\end{tikzpicture}
\caption{Best response of Sally and Bob in the first few iterations}
\label{fig:spy_br_1}
\end{figure}
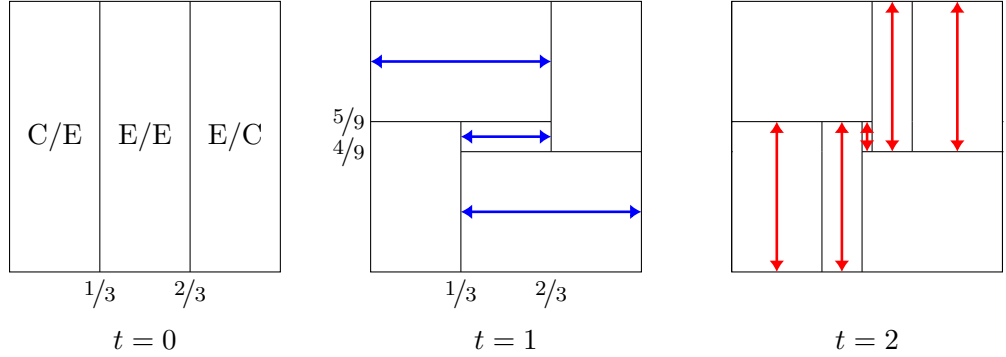
}

\newcommand{\SpyTOne}{
\begin{figure}[h]
\centering
\begin{tikzpicture}[scale=4]

 \begin{scope}[xshift=0cm, scale=.9, yscale=.4]
\draw (0,0) -- (1,0);
\draw (0,-1.555556) -- (0,1.666667);
\draw [line width=1.5pt, color=blue] (0.000000, 0.000000)--(0.333333, 0.333333)--(0.333333, -1.444444)--(0.666667, -1.555556)--(0.666667, -0.666667)--(1.000000, -1.000000);
\draw [line width=1.5pt, color=blue, dashed]  (0.333333, 0.333333)--(1.000000, -1.000000);
\draw [line width=1.5pt, color=red] (0.000000, 1.333333)--(0.333333, 0.333333)--(0.333333, 0.333333)--(0.666667, 0.666667)--(0.666667, 0.666667)--(1.000000, 1.666667);
\node at (.95,.-.6) {$B$};
\node at (.95,1.25) {$S$};
\node at (1/2, -2) {$q = 1/3$};
 \end{scope}
 
\begin{scope}[xshift=1.2cm, scale=.9, yscale=.4]
\draw (0,0) -- (1,0);
\draw (0,-1.500000) -- (0,1.500000);
\draw [line width=1.5pt, color=blue] (0.000000, -0.500000)--(0.333333, -0.166667)--(0.333333, -1.500000)--(0.666667, -1.500000)--(0.666667, -0.166667)--(1.000000, -0.500000);
\draw [line width=1.5pt, color=blue, dashed]  (0.333333, -0.166667)--(0.666667, -0.166667);
\draw [line width=1.5pt, color=red] (0.000000, 1.500000)--(0.333333, 0.500000)--(0.333333, 0.500000)--(0.666667, 0.500000)--(0.666667, 0.500000)--(1.000000, 1.500000);
\node at (.95,.-.16) {$B$};
\node at (.95,1.1) {$S$};
\node at (1/2, -2) {$q = 1/2$};
 \end{scope}
 
\begin{scope}[xshift=2.4cm, scale=.9, yscale=.4]
\draw (0,0) -- (1,0);
\draw (0,-1.555556) -- (0,1.666667);
\draw [line width=1.5pt, color=blue] (0.000000, -1.000000)--(0.333333, -0.666667)--(0.333333, -1.555556)--(0.666667, -1.444444)--(0.666667, 0.333333)--(1.000000, 0.000000);
\draw [line width=1.5pt, color=blue, dashed]  (0.000000, -1.000000)--(0.666667, 0.333333);
\draw [line width=1.5pt, color=red] (0.000000, 1.666667)--(0.333333, 0.666667)--(0.333333, 0.666667)--(0.666667, 0.333333)--(0.666667, 0.333333)--(1.000000, 1.333333);
\node at (.95,.2) {$B$};
\node at (.95,1) {$S$};
\node at (1/2, -2) {$q = 2/3$};
 \end{scope}

\end{tikzpicture}
\caption{Bob and Sally's payoffs as a function of $p$ at $t=0$}
\label{fig:spy_01}
\end{figure}
}

\newcommand{\SpyBestResponseTwo}{
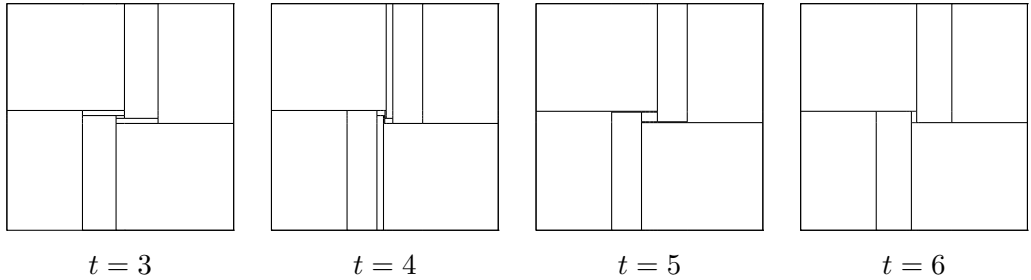
\begin{figure}[h]
\centering
\begin{tikzpicture}[scale=3.35] 

 \begin{scope}[xshift=0cm, scale=.9]
 \draw (0.000000,0.000000)--(0.000000,0.000000);
\draw (0.000000,0.000000)--(0.000000,0.471380);
\draw (0.000000,0.471380)--(0.000000,0.493827);
\draw (0.000000,0.493827)--(0.000000,0.506173);
\draw (0.000000,0.506173)--(0.000000,0.528620);
\draw (0.000000,0.528620)--(0.000000,1.000000);
\draw (0.000000,1.000000)--(0.000000,1.000000);
\draw (0.333333,0.000000)--(0.333333,0.000000);
\draw (0.333333,0.000000)--(0.333333,0.471380);
\draw (0.333333,0.471380)--(0.333333,0.493827);
\draw (0.333333,0.493827)--(0.333333,0.506173);
\draw (0.333333,0.506173)--(0.333333,0.528620);
\draw (0.333333,1.000000)--(0.333333,1.000000);
\draw (0.481481,0.000000)--(0.481481,0.000000);
\draw (0.481481,0.000000)--(0.481481,0.471380);
\draw (0.481481,0.471380)--(0.481481,0.493827);
\draw (0.481481,0.493827)--(0.481481,0.506173);
\draw (0.481481,1.000000)--(0.481481,1.000000);
\draw (0.518519,0.000000)--(0.518519,0.000000);
\draw (0.518519,0.493827)--(0.518519,0.506173);
\draw (0.518519,0.506173)--(0.518519,0.528620);
\draw (0.518519,0.528620)--(0.518519,1.000000);
\draw (0.518519,1.000000)--(0.518519,1.000000);
\draw (0.666667,0.471380)--(0.666667,0.493827);
\draw (0.666667,0.493827)--(0.666667,0.506173);
\draw (0.666667,0.506173)--(0.666667,0.528620);
\draw (0.666667,0.528620)--(0.666667,1.000000);
\draw (0.666667,1.000000)--(0.666667,1.000000);
\draw (1.000000,0.000000)--(1.000000,0.000000);
\draw (1.000000,0.000000)--(1.000000,0.471380);
\draw (1.000000,0.471380)--(1.000000,0.493827);
\draw (1.000000,0.493827)--(1.000000,0.506173);
\draw (1.000000,0.506173)--(1.000000,0.528620);
\draw (1.000000,0.528620)--(1.000000,1.000000);
\draw (1.000000,1.000000)--(1.000000,1.000000);
\draw (0.333333,0.000000)--(0.481481,0.000000);
\draw (0.481481,0.000000)--(0.518519,0.000000);
\draw (0.518519,0.000000)--(0.666667,0.000000);
\draw (0.666667,0.000000)--(1.000000,0.000000);
\draw (0.481481,0.471380)--(0.518519,0.471380);
\draw (0.518519,0.471380)--(0.666667,0.471380);
\draw (0.666667,0.471380)--(1.000000,0.471380);
\draw (0.481481,0.493827)--(0.518519,0.493827);
\draw (0.518519,0.493827)--(0.666667,0.493827);
\draw (0.333333,0.506173)--(0.481481,0.506173);
\draw (0.481481,0.506173)--(0.518519,0.506173);
\draw (0.000000,0.528620)--(0.333333,0.528620);
\draw (0.333333,0.528620)--(0.481481,0.528620);
\draw (0.481481,0.528620)--(0.518519,0.528620);
\draw (0.000000,1.000000)--(0.333333,1.000000);
\draw (0.333333,1.000000)--(0.481481,1.000000);
\draw (0.481481,1.000000)--(0.518519,1.000000);
\draw (0.518519,1.000000)--(0.666667,1.000000);
\draw (0,0) -- (1,0) -- (1,1) -- (0,1) -- cycle;

    \node at (1/2, -.15) {$t=3$};
 \end{scope}
 
  \begin{scope}[xshift=1.05cm, scale=.9]
  \draw (0.000000,0.000000)--(0.000000,0.000000);
\draw (0.000000,0.000000)--(0.000000,0.471380);
\draw (0.000000,0.471380)--(0.000000,0.493827);
\draw (0.000000,0.493827)--(0.000000,0.506173);
\draw (0.000000,0.506173)--(0.000000,0.528620);
\draw (0.000000,0.528620)--(0.000000,1.000000);
\draw (0.000000,1.000000)--(0.000000,1.000000);
\draw (0.333333,0.000000)--(0.333333,0.000000);
\draw (0.333333,0.000000)--(0.333333,0.471380);
\draw (0.333333,0.471380)--(0.333333,0.493827);
\draw (0.333333,0.493827)--(0.333333,0.506173);
\draw (0.333333,0.506173)--(0.333333,0.528620);
\draw (0.464997,0.000000)--(0.464997,0.471380);
\draw (0.464997,0.471380)--(0.464997,0.493827);
\draw (0.464997,0.493827)--(0.464997,0.506173);
\draw (0.464997,0.506173)--(0.464997,0.528620);
\draw (0.494023,0.000000)--(0.494023,0.471380);
\draw (0.494023,0.471380)--(0.494023,0.493827);
\draw (0.494023,0.493827)--(0.494023,0.506173);
\draw (0.499277,0.471380)--(0.499277,0.493827);
\draw (0.499277,0.493827)--(0.499277,0.506173);
\draw (0.500723,0.493827)--(0.500723,0.506173);
\draw (0.500723,0.506173)--(0.500723,0.528620);
\draw (0.505977,0.493827)--(0.505977,0.506173);
\draw (0.505977,0.506173)--(0.505977,0.528620);
\draw (0.505977,0.528620)--(0.505977,1.000000);
\draw (0.535003,0.471380)--(0.535003,0.493827);
\draw (0.535003,0.493827)--(0.535003,0.506173);
\draw (0.535003,0.506173)--(0.535003,0.528620);
\draw (0.535003,0.528620)--(0.535003,1.000000);
\draw (0.666667,0.471380)--(0.666667,0.493827);
\draw (0.666667,0.493827)--(0.666667,0.506173);
\draw (0.666667,0.506173)--(0.666667,0.528620);
\draw (0.666667,0.528620)--(0.666667,1.000000);
\draw (0.666667,1.000000)--(0.666667,1.000000);
\draw (1.000000,0.000000)--(1.000000,0.000000);
\draw (1.000000,0.000000)--(1.000000,0.471380);
\draw (1.000000,0.471380)--(1.000000,0.493827);
\draw (1.000000,0.493827)--(1.000000,0.506173);
\draw (1.000000,0.506173)--(1.000000,0.528620);
\draw (1.000000,0.528620)--(1.000000,1.000000);
\draw (1.000000,1.000000)--(1.000000,1.000000);
\draw (0.000000,0.000000)--(0.000000,0.000000);
\draw (0.000000,0.000000)--(0.333333,0.000000);
\draw (0.333333,0.000000)--(0.464997,0.000000);
\draw (0.464997,0.000000)--(0.481481,0.000000);
\draw (0.481481,0.000000)--(0.494023,0.000000);
\draw (0.494023,0.000000)--(0.499277,0.000000);
\draw (0.499277,0.000000)--(0.500723,0.000000);
\draw (0.500723,0.000000)--(0.505977,0.000000);
\draw (0.505977,0.000000)--(0.518519,0.000000);
\draw (0.518519,0.000000)--(0.535003,0.000000);
\draw (0.535003,0.000000)--(0.666667,0.000000);
\draw (0.666667,0.000000)--(1.000000,0.000000);
\draw (1.000000,0.000000)--(1.000000,0.000000);
\draw (0.494023,0.471380)--(0.499277,0.471380);
\draw (0.499277,0.471380)--(0.500723,0.471380);
\draw (0.500723,0.471380)--(0.505977,0.471380);
\draw (0.505977,0.471380)--(0.518519,0.471380);
\draw (0.518519,0.471380)--(0.535003,0.471380);
\draw (0.535003,0.471380)--(0.666667,0.471380);
\draw (0.666667,0.471380)--(1.000000,0.471380);
\draw (0.499277,0.493827)--(0.500723,0.493827);
\draw (0.500723,0.493827)--(0.505977,0.493827);
\draw (0.505977,0.493827)--(0.518519,0.493827);
\draw (0.518519,0.493827)--(0.535003,0.493827);
\draw (0.464997,0.506173)--(0.481481,0.506173);
\draw (0.481481,0.506173)--(0.494023,0.506173);
\draw (0.494023,0.506173)--(0.499277,0.506173);
\draw (0.499277,0.506173)--(0.500723,0.506173);
\draw (0.000000,0.528620)--(0.333333,0.528620);
\draw (0.333333,0.528620)--(0.464997,0.528620);
\draw (0.464997,0.528620)--(0.481481,0.528620);
\draw (0.481481,0.528620)--(0.494023,0.528620);
\draw (0.494023,0.528620)--(0.499277,0.528620);
\draw (0.499277,0.528620)--(0.500723,0.528620);
\draw (0.500723,0.528620)--(0.505977,0.528620);
\draw (0.000000,1.000000)--(0.000000,1.000000);
\draw (0.000000,1.000000)--(0.333333,1.000000);
\draw (0.333333,1.000000)--(0.464997,1.000000);
\draw (0.464997,1.000000)--(0.481481,1.000000);
\draw (0.481481,1.000000)--(0.494023,1.000000);
\draw (0.494023,1.000000)--(0.499277,1.000000);
\draw (0.499277,1.000000)--(0.500723,1.000000);
\draw (0.500723,1.000000)--(0.505977,1.000000);
\draw (0.505977,1.000000)--(0.518519,1.000000);
\draw (0.518519,1.000000)--(0.535003,1.000000);
\draw (0.535003,1.000000)--(0.666667,1.000000);
\draw (0.666667,1.000000)--(1.000000,1.000000);
\draw (1.000000,1.000000)--(1.000000,1.000000);
\draw (0,0) -- (1,0) -- (1,1) -- (0,1) -- cycle;

    \node at (1/2, -.15) {$t=4$};
 \end{scope}

 \begin{scope}[xshift=2.1cm, scale=.9]
 \draw (0.000000,0.000000)--(0.000000,0.000000);
\draw (0.000000,0.000000)--(0.000000,0.474970);
\draw (0.000000,0.474970)--(0.000000,0.478676);
\draw (0.000000,0.478676)--(0.000000,0.521324);
\draw (0.000000,0.521324)--(0.000000,0.525030);
\draw (0.000000,0.525030)--(0.000000,1.000000);
\draw (0.000000,1.000000)--(0.000000,1.000000);
\draw (0.333333,0.000000)--(0.333333,0.000000);
\draw (0.333333,0.000000)--(0.333333,0.474970);
\draw (0.333333,0.474970)--(0.333333,0.478676);
\draw (0.333333,0.478676)--(0.333333,0.521324);
\draw (0.333333,0.521324)--(0.333333,0.525030);
\draw (0.464997,0.000000)--(0.464997,0.474970);
\draw (0.464997,0.474970)--(0.464997,0.478676);
\draw (0.464997,0.478676)--(0.464997,0.521324);
\draw (0.535003,0.478676)--(0.535003,0.521324);
\draw (0.535003,0.521324)--(0.535003,0.525030);
\draw (0.535003,0.525030)--(0.535003,1.000000);
\draw (0.666667,0.474970)--(0.666667,0.478676);
\draw (0.666667,0.478676)--(0.666667,0.521324);
\draw (0.666667,0.521324)--(0.666667,0.525030);
\draw (0.666667,0.525030)--(0.666667,1.000000);
\draw (0.666667,1.000000)--(0.666667,1.000000);
\draw (1.000000,0.000000)--(1.000000,0.000000);
\draw (1.000000,0.000000)--(1.000000,0.474970);
\draw (1.000000,0.474970)--(1.000000,0.478676);
\draw (1.000000,0.478676)--(1.000000,0.521324);
\draw (1.000000,0.521324)--(1.000000,0.525030);
\draw (1.000000,0.525030)--(1.000000,1.000000);
\draw (1.000000,1.000000)--(1.000000,1.000000);
\draw (0.000000,0.000000)--(0.000000,0.000000);
\draw (0.000000,0.000000)--(0.333333,0.000000);
\draw (0.333333,0.000000)--(0.464997,0.000000);
\draw (0.464997,0.000000)--(0.481481,0.000000);
\draw (0.481481,0.000000)--(0.494023,0.000000);
\draw (0.494023,0.000000)--(0.499277,0.000000);
\draw (0.499277,0.000000)--(0.500723,0.000000);
\draw (0.500723,0.000000)--(0.505977,0.000000);
\draw (0.505977,0.000000)--(0.518519,0.000000);
\draw (0.518519,0.000000)--(0.535003,0.000000);
\draw (0.535003,0.000000)--(0.666667,0.000000);
\draw (0.666667,0.000000)--(1.000000,0.000000);
\draw (1.000000,0.000000)--(1.000000,0.000000);
\draw (0.464997,0.474970)--(0.481481,0.474970);
\draw (0.481481,0.474970)--(0.494023,0.474970);
\draw (0.494023,0.474970)--(0.499277,0.474970);
\draw (0.499277,0.474970)--(0.500723,0.474970);
\draw (0.500723,0.474970)--(0.505977,0.474970);
\draw (0.505977,0.474970)--(0.518519,0.474970);
\draw (0.518519,0.474970)--(0.535003,0.474970);
\draw (0.535003,0.474970)--(0.666667,0.474970);
\draw (0.666667,0.474970)--(1.000000,0.474970);
\draw (0.464997,0.478676)--(0.481481,0.478676);
\draw (0.481481,0.478676)--(0.494023,0.478676);
\draw (0.494023,0.478676)--(0.499277,0.478676);
\draw (0.499277,0.478676)--(0.500723,0.478676);
\draw (0.500723,0.478676)--(0.505977,0.478676);
\draw (0.505977,0.478676)--(0.518519,0.478676);
\draw (0.518519,0.478676)--(0.535003,0.478676);
\draw (0.535003,0.478676)--(0.666667,0.478676);
\draw (0.333333,0.521324)--(0.464997,0.521324);
\draw (0.464997,0.521324)--(0.481481,0.521324);
\draw (0.481481,0.521324)--(0.494023,0.521324);
\draw (0.494023,0.521324)--(0.499277,0.521324);
\draw (0.499277,0.521324)--(0.500723,0.521324);
\draw (0.500723,0.521324)--(0.505977,0.521324);
\draw (0.505977,0.521324)--(0.518519,0.521324);
\draw (0.518519,0.521324)--(0.535003,0.521324);
\draw (0.000000,0.525030)--(0.333333,0.525030);
\draw (0.333333,0.525030)--(0.464997,0.525030);
\draw (0.464997,0.525030)--(0.481481,0.525030);
\draw (0.481481,0.525030)--(0.494023,0.525030);
\draw (0.494023,0.525030)--(0.499277,0.525030);
\draw (0.499277,0.525030)--(0.500723,0.525030);
\draw (0.500723,0.525030)--(0.505977,0.525030);
\draw (0.505977,0.525030)--(0.518519,0.525030);
\draw (0.518519,0.525030)--(0.535003,0.525030);
\draw (0.000000,1.000000)--(0.000000,1.000000);
\draw (0.000000,1.000000)--(0.333333,1.000000);
\draw (0.333333,1.000000)--(0.464997,1.000000);
\draw (0.464997,1.000000)--(0.481481,1.000000);
\draw (0.481481,1.000000)--(0.494023,1.000000);
\draw (0.494023,1.000000)--(0.499277,1.000000);
\draw (0.499277,1.000000)--(0.500723,1.000000);
\draw (0.500723,1.000000)--(0.505977,1.000000);
\draw (0.505977,1.000000)--(0.518519,1.000000);
\draw (0.518519,1.000000)--(0.535003,1.000000);
\draw (0.535003,1.000000)--(0.666667,1.000000);
\draw (0.666667,1.000000)--(1.000000,1.000000);
\draw (1.000000,1.000000)--(1.000000,1.000000);
\draw (0,0) -- (1,0) -- (1,1) -- (0,1) -- cycle;

    \node at (1/2, -.15) {$t=5$};
 \end{scope}

 \begin{scope}[xshift=3.15cm, scale=.9]
\draw (0.000000,0.000000)--(0.000000,0.000000);
\draw (0.000000,0.000000)--(0.000000,0.474970);
\draw (0.000000,0.474970)--(0.000000,0.525030);
\draw (0.000000,0.525030)--(0.000000,1.000000);
\draw (0.000000,1.000000)--(0.000000,1.000000);
\draw (0.333333,0.000000)--(0.333333,0.000000);
\draw (0.333333,0.000000)--(0.333333,0.474970);
\draw (0.333333,0.474970)--(0.333333,0.525030);
\draw (0.488354,0.000000)--(0.488354,0.474970);
\draw (0.488354,0.474970)--(0.488354,0.525030);
\draw (0.511646,0.474970)--(0.511646,0.525030);
\draw (0.511646,0.525030)--(0.511646,1.000000);
\draw (0.666667,0.474970)--(0.666667,0.525030);
\draw (0.666667,0.525030)--(0.666667,1.000000);
\draw (0.666667,1.000000)--(0.666667,1.000000);
\draw (1.000000,0.000000)--(1.000000,0.000000);
\draw (1.000000,0.000000)--(1.000000,0.474970);
\draw (1.000000,0.474970)--(1.000000,0.525030);
\draw (1.000000,0.525030)--(1.000000,1.000000);
\draw (1.000000,1.000000)--(1.000000,1.000000);
\draw (0.000000,0.000000)--(0.000000,0.000000);
\draw (0.000000,0.000000)--(0.333333,0.000000);
\draw (0.333333,0.000000)--(0.464997,0.000000);
\draw (0.464997,0.000000)--(0.481481,0.000000);
\draw (0.481481,0.000000)--(0.488354,0.000000);
\draw (0.488354,0.000000)--(0.494023,0.000000);
\draw (0.494023,0.000000)--(0.499277,0.000000);
\draw (0.499277,0.000000)--(0.500723,0.000000);
\draw (0.500723,0.000000)--(0.505977,0.000000);
\draw (0.505977,0.000000)--(0.511646,0.000000);
\draw (0.511646,0.000000)--(0.518519,0.000000);
\draw (0.518519,0.000000)--(0.535003,0.000000);
\draw (0.535003,0.000000)--(0.666667,0.000000);
\draw (0.666667,0.000000)--(1.000000,0.000000);
\draw (1.000000,0.000000)--(1.000000,0.000000);
\draw (0.488354,0.474970)--(0.494023,0.474970);
\draw (0.494023,0.474970)--(0.499277,0.474970);
\draw (0.499277,0.474970)--(0.500723,0.474970);
\draw (0.500723,0.474970)--(0.505977,0.474970);
\draw (0.505977,0.474970)--(0.511646,0.474970);
\draw (0.511646,0.474970)--(0.518519,0.474970);
\draw (0.518519,0.474970)--(0.535003,0.474970);
\draw (0.535003,0.474970)--(0.666667,0.474970);
\draw (0.666667,0.474970)--(1.000000,0.474970);
\draw (0.000000,0.525030)--(0.333333,0.525030);
\draw (0.333333,0.525030)--(0.464997,0.525030);
\draw (0.464997,0.525030)--(0.481481,0.525030);
\draw (0.481481,0.525030)--(0.488354,0.525030);
\draw (0.488354,0.525030)--(0.494023,0.525030);
\draw (0.494023,0.525030)--(0.499277,0.525030);
\draw (0.499277,0.525030)--(0.500723,0.525030);
\draw (0.500723,0.525030)--(0.505977,0.525030);
\draw (0.505977,0.525030)--(0.511646,0.525030);
\draw (0.000000,1.000000)--(0.000000,1.000000);
\draw (0.000000,1.000000)--(0.333333,1.000000);
\draw (0.333333,1.000000)--(0.464997,1.000000);
\draw (0.464997,1.000000)--(0.481481,1.000000);
\draw (0.481481,1.000000)--(0.488354,1.000000);
\draw (0.488354,1.000000)--(0.494023,1.000000);
\draw (0.494023,1.000000)--(0.499277,1.000000);
\draw (0.499277,1.000000)--(0.500723,1.000000);
\draw (0.500723,1.000000)--(0.505977,1.000000);
\draw (0.505977,1.000000)--(0.511646,1.000000);
\draw (0.511646,1.000000)--(0.518519,1.000000);
\draw (0.518519,1.000000)--(0.535003,1.000000);
\draw (0.535003,1.000000)--(0.666667,1.000000);
\draw (0.666667,1.000000)--(1.000000,1.000000);
\draw (1.000000,1.000000)--(1.000000,1.000000);
\draw (0,0) -- (1,0) -- (1,1) -- (0,1) -- cycle;
    \node at (1/2, -.15) {$t=6$};
 \end{scope}

\end{tikzpicture}
\caption{Best response of Sally and Bob in iterations $t=3,4,5,6$}
\label{fig:spy_br_2}
\end{figure}
}

\newcommand{\SpyInfoPathTwo}{
\begin{figure}[h]
\centering
\begin{tikzpicture}[scale=8]

\draw[\arw-\arw, line width=1pt, color=red] (0.500000,0.444444)--(0.500000,0.555556);
\draw[\arw-\arw, line width=1pt, color=blue] (0.333333,0.444444)--(1.000000,0.444444);
\draw[\arw-\arw, line width=1pt, color=blue] (0.000000,0.555556)--(0.666667,0.555556);
\node[circle,fill,inner sep=1pt] at (0.500000,0.500000) {};
\node[circle,fill,inner sep=1pt] at (0.500000,0.444444) {};
\node[circle,fill,inner sep=1pt] at (0.500000,0.555556) {};
\draw (0,0) -- (1,0) -- (1,1) -- (0,1) -- cycle;

\end{tikzpicture}
\caption{Best response of Sally and Bob in iterations $t=3,4,5$}
\label{fig:spy_t345}
\end{figure}
}

\newcommand{\SpyInfoPathThree}{
\begin{figure}[h]
\centering
\begin{tikzpicture}[scale=8]

\draw[\arw-\arw, line width=1pt, color=blue] (0.481481,0.500000)--(0.518519,0.500000);
\draw[\arw-\arw, line width=1pt, color=red] (0.481481,0.000000)--(0.481481,0.555556);
\draw[\arw-\arw, line width=1pt, color=red] (0.518519,0.444444)--(0.518519,1.000000);
\draw[\arw-\arw, line width=1pt, color=blue] (0.333333,0.000000)--(1.000000,0.000000);
\draw[\arw-\arw, line width=1pt, color=blue] (0.000000,0.555556)--(0.666667,0.555556);
\draw[\arw-\arw, line width=1pt, color=blue] (0.333333,0.444444)--(1.000000,0.444444);
\draw[\arw-\arw, line width=1pt, color=blue] (0.000000,1.000000)--(0.666667,1.000000);
\node[circle,fill,inner sep=1pt] at (0.500000,0.500000) {};
\node[circle,fill,inner sep=1pt] at (0.481481,0.500000) {};
\node[circle,fill,inner sep=1pt] at (0.518519,0.500000) {};
\node[circle,fill,inner sep=1pt] at (0.481481,0.000000) {};
\node[circle,fill,inner sep=1pt] at (0.481481,0.555556) {};
\node[circle,fill,inner sep=1pt] at (0.518519,0.444444) {};
\node[circle,fill,inner sep=1pt] at (0.518519,1.000000) {};
\node[circle,fill,inner sep=1pt] at (0.333333,0.000000) {};
\node[circle,fill,inner sep=1pt] at (1.000000,0.000000) {};
\node[circle,fill,inner sep=1pt] at (0.000000,0.555556) {};
\node[circle,fill,inner sep=1pt] at (0.666667,0.555556) {};
\node[circle,fill,inner sep=1pt] at (0.333333,0.444444) {};
\node[circle,fill,inner sep=1pt] at (1.000000,0.444444) {};
\node[circle,fill,inner sep=1pt] at (0.000000,1.000000) {};
\node[circle,fill,inner sep=1pt] at (0.666667,1.000000) {};
\draw (0,0) -- (1,0) -- (1,1) -- (0,1) -- cycle;

\end{tikzpicture}
\caption{Best response of Sally and Bob in iterations $t=3,4,5$}
\label{fig:spy_t345}
\end{figure}
}

\newcommand{\SpyInfoPathFive}{
\begin{figure}[h]
\centering
\begin{tikzpicture}[scale=7]
\draw[\arw-\arw, line width=1pt, color=blue] (0.464997,0.500000)--(0.535003,0.500000);
\draw[\arw-\arw, line width=1pt, color=red] (0.464997,0.000000)--(0.464997,0.528620);
\draw[\arw-\arw, line width=1pt, color=red] (0.535003,0.471380)--(0.535003,1.000000);
\draw[\arw-\arw, line width=1pt, color=blue] (0.000000,0.528620)--(0.518519,0.528620);
\draw[\arw-\arw, line width=1pt, color=blue] (0.481481,0.471380)--(1.000000,0.471380);
\draw[\arw-\arw, line width=1pt, color=red] (0.518519,0.444444)--(0.518519,1.000000);
\draw[\arw-\arw, line width=1pt, color=red] (0.481481,0.000000)--(0.481481,0.555556);
\draw[\arw-\arw, line width=1pt, color=blue] (0.333333,0.444444)--(1.000000,0.444444);
\draw[\arw-\arw, line width=1pt, color=blue] (0.000000,1.000000)--(0.666667,1.000000);
\draw[\arw-\arw, line width=1pt, color=blue] (0.333333,0.000000)--(1.000000,0.000000);
\draw[\arw-\arw, line width=1pt, color=blue] (0.000000,0.555556)--(0.666667,0.555556);
\node[circle,fill,inner sep=1pt] at (0.500000,0.500000) {};
\node[circle,fill,inner sep=1pt] at (0.464997,0.500000) {};
\node[circle,fill,inner sep=1pt] at (0.535003,0.500000) {};
\node[circle,fill,inner sep=1pt] at (0.464997,0.000000) {};
\node[circle,fill,inner sep=1pt] at (0.464997,0.528620) {};
\node[circle,fill,inner sep=1pt] at (0.535003,0.471380) {};
\node[circle,fill,inner sep=1pt] at (0.535003,1.000000) {};
\node[circle,fill,inner sep=1pt] at (0.000000,0.528620) {};
\node[circle,fill,inner sep=1pt] at (0.518519,0.528620) {};
\node[circle,fill,inner sep=1pt] at (0.481481,0.471380) {};
\node[circle,fill,inner sep=1pt] at (1.000000,0.471380) {};
\node[circle,fill,inner sep=1pt] at (0.518519,0.444444) {};
\node[circle,fill,inner sep=1pt] at (0.518519,1.000000) {};
\node[circle,fill,inner sep=1pt] at (0.481481,0.000000) {};
\node[circle,fill,inner sep=1pt] at (0.481481,0.555556) {};
\node[circle,fill,inner sep=1pt] at (0.333333,0.444444) {};
\node[circle,fill,inner sep=1pt] at (1.000000,0.444444) {};
\node[circle,fill,inner sep=1pt] at (0.000000,1.000000) {};
\node[circle,fill,inner sep=1pt] at (0.666667,1.000000) {};
\node[circle,fill,inner sep=1pt] at (0.333333,0.000000) {};
\node[circle,fill,inner sep=1pt] at (1.000000,0.000000) {};
\node[circle,fill,inner sep=1pt] at (0.000000,0.555556) {};
\node[circle,fill,inner sep=1pt] at (0.666667,0.555556) {};
\draw (0,0) -- (1,0) -- (1,1) -- (0,1) -- cycle;

\end{tikzpicture}
\caption{Dynamics of the state $(p,q)$ starting at $(\nicefrac{1}{2},\nicefrac{1}{2})$ for $5$ periods}
\label{fig:spy_dynamic_5}
\end{figure}
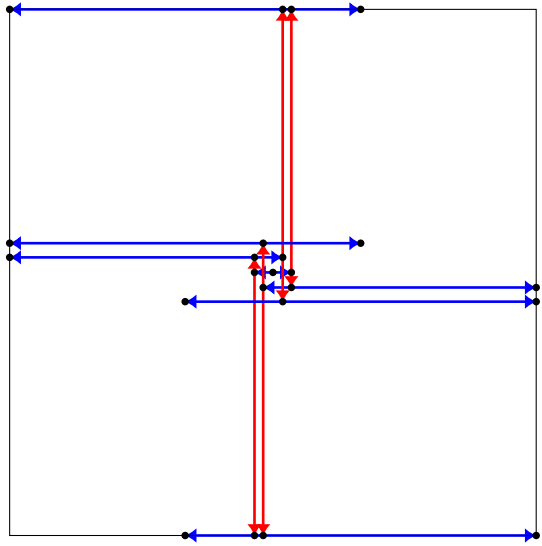
}

\section{A Spy Game (with Longer Communication)} \label{sec:spy}

We conclude by showing an example of a game with a longer communication protocol. In the tradition of Cold War era game theory we will have Sally (the Spymaster) and Bob (the Birdwatcher) be spies. As with good spies, no one really knows which country they are actually serving. Their types $\theta_S, \theta_B \in \{0,1\}$ represent the country they are serving. If $\theta_S = \theta_B$ then Sally and Bob are friends and should cooperate. If $\theta_S \neq \theta_B$ they are enemies and they have opposing goals. In their meeting in West Berlin (Glienicke Bridge, more precisely), Sally must decide either to cooperate with Bob (C) or to expose him (E). The actions and types lead to the following utilities for Sally and Bob respectively in the base game:

\begin{center}
\begin{tabular}{ c|c|c } 
  & C & E  \\ 
  \hline
  $\theta_S = \theta_B$ & $(1,1)$ & $(-1,-1)$  \\ 
  $\theta_S \neq \theta_B$ & $(-2,2)$ & $(2,-2)$  \\ 
\end{tabular}
\end{center}

If Sally and Bob are friends ($\theta_S = \theta_B$) then they are playing a cooperative game: both get a payoff of $1$ for cooperating and both get $-1$ if Sally decides to expose Bob. If Sally and Bob are enemies ($\theta_S \neq \theta_B$) they are playing a zero-sum game. If Sally cooperates with an enemy, she loses ($-2$) and Bob wins ($+2$). If she exposes an enemy then Sally wins ($+2$) and Bob loses ($-2$).

We analyze this game by following the same methodology as in Section \ref{sec:binary_types} of taking alternate concave hulls in Bob's and Sally's dimensions (see implementation in footnote 1). Even though it is a game with binary types and only two actions, the communication protocol arising in equilibrium is quite complex. For example, the following table shows the payoffs of Sally and Bob from a protocol with $t$ periods if they start from uniform distributions:

\begin{center}
\begin{tabular}{ c|c|c|c|c|c|c|c } 
  & $t=0$ & $t=1$ & $t=2$ & $t=3$ & $t=4$ & $t=5$ & $t=6$ \\ 
  \hline
  S & 0.5 & 0.5 & 0.722 & 0.722 & 0.738 & 0.769 & 0.801 \\ 
  B & -1.5 & -0.166 & -0.166 & -0.107 & -0.107 & -0.077 & -0.075  \\ 
\end{tabular}
\end{center}

In Figures \ref{fig:spy_br_1} and \ref{fig:spy_br_2} we depict the optimal protocol using the square diagrams similar to the ones introduced in Figure \ref{fig:two_types_example}. The coordinates of each point $(p,q) \in [0,1]^2$ represent the probabilities that Bob ($p$) and Sally ($q$) have type $1$.

We note that the optimal strategy can be represented simply by a partition of the square $[0,1]^2$ into sub-rectangles. Assume we are in state $(p,q)$ in round $t$ and its inside a sub-rectangle $[p_0, p_1] \times [q_0, q_1]$. If $t$ is odd and Bob will refine $(p,q)$ to $(p_0,q)$ and $(p_1,q)$. If $t$ is even, Sally  will refine $(p,q)$ to $(p, q_0)$ and $(p,q_1)$.

At $t=0$ Sally's action varies whether her belief is in the intervals $[0,1/3]$, $(1/3,2/3)$ or $[2/3, 1]$. In the middle interval she always exposed Bob and in the first and third interval she will cooperate if Bob's likeliest type is the same as her type. 

\SpyBestResponse
\SpyBestResponseTwo

What follows next is a delicate game where agents slowly reveal their types depending on how much information they have about other agents. For example, Bob's strategy in $t=1$ will be one of the following:
\begin{itemize}
    \item If Bob has a more extreme information about Sally ($q \leq 4/9$ or $q \geq 5/9$) he will reveal more information about his type hoping to cooperate more frequently.
    \item If Bob is more uncertain about Sally's type ($4/9 < q < 5/9$) Bob will reveal less information since he runs the risk of Sally having the opposite type and deciding to expose him when she would otherwise cooperate.
\end{itemize}
In $t=2$ Sally is again in the same situation: the more certain she is about Bob's type the more willing she is to tell her own type.
 Figure \ref{fig:spy_br_2} shows the evolution of the game in each of the subsequent stages. 

\SpyInfoPathFive

Figure \ref{fig:spy_dynamic_5} shows how the state $(p,q)$ of the game evolves starting from $(1/2,1/2)$ for $5$ iterations. What we observe is that a long and gradual disclose of information benefits both parties.


\bibliographystyle{plain}
\bibliography{info_design}

\appendix
\section{Missing proofs for generic protocols}
\label{sec:reductions}
We show here that it's sufficient to only consider alternating information refinement protocols (defined in Section \ref{subsubsec:msg_struct}) among a more general set of protocols.

\subsection{Alternating refinement vs non-alternating refinement}

We first define a more general version of the alternating information refinement protocols which does not require alternating between Sally and Bob, and we call them information refinement protocols. More specifically, in an information refinement protocol, in each round $t$, both parties are allowed to refine information, i.e.
$$\begin{aligned}
 & & \D_B^{t-1} \sim m_B^t(\D_S^t, \D_B^t) & ,\quad  \D_S^{t-1} \sim m_S^t(\D_S^t, \D_B^t).
\end{aligned}$$

 We show that it has the same utilities as an alternating one in which each party uses an additional input of the other party's previous round: i.e. $\D_B^{t-1} \sim m_B^t(\D_S^t, \D_B^t, \D_S^{t+1})$ or $\D_S^{t-1} \sim m_S^t(\D_S^t, \D_B^t, \D_B^{t+1})$. Notice that in an equilibrium, this additional input of the previous round can be dropped.
\begin{lemma}
For any (not necessarily alternating) information refinement protocol $P$, there exists an alternating information refinement protocol $P'$ with the same utilities.
\end{lemma}

\begin{proof}
Let the refinement of $P$ in each round $t$ to be $m_B^t(d_S, d_B)$ and $m_S^t(d_S, d_B)$ for any $d_S, d_B$. We define the refinement of $P'$ to be $n_S^t$ and $n_B^t$ (for notational convenience, $m_S^{T+1}(d_S,d_B) = d_S$ and $m_B^{T+1}(d_S,d_B) = d_B$): 
\begin{itemize}
\item When $t$ is odd, $n_S^t(d_S,d_B,d'_B) = d_S$, and $n_B^t(d_S,d_B,d'_S)= m_B^{t}(d_S,m_B^{t+1}(d'_S,d_B))$ for any $d_S, d_B, d'_S, d'_B$.
\item When $t$ is even, $n_S^t(d_S,d_B,d'_B) = m_S^{t}(m_S^{t+1}(d_S,d_B'),d_B)$, and $n_B^t(d_S,d_B,d'_S)= d_B$ for any $d_S, d_B$. 
\end{itemize}
Clearly, $P'$ is alternating. We prove the equivalence between utilities by a hybrid argument. Consider refinement protocol $P_{\tau}$ to be the one such that 
\begin{itemize}
    \item For round $t <\tau$, use refinement $n_S^t$ and $n_B^t$.
    \item For round $t = \tau$, if $\tau$ is odd use $n_S^t$ and $m_B^t$. Otherwise, use $m_S^t$ and $n_B^t$.
    \item For round $t >\tau$, use refinement $m_S^t$ and $m_B^t$.
\end{itemize}
It's easy to see that $P_0 = P$ and $P_{T+1} = P'$. Now it suffices to show that $P_t$ and $P_{t+1}$ have the same utilities for $t = 0,\ldots, T+1$.
\begin{itemize}
    \item For $P_0$ versus $P_1$, it's easy to see that the last message sent by Sally ($m^1_S$) in $P_0$ does not affect the action taken by Sally, and therefore utilities are not affected.
    \item For $P_t$ versus $P_{t+1}$ with odd $t$, they only differ in rounds $t$ and $t+1$. $P_t$ has refinements $m_S^{t+1}$ and $m_B^{t+1}$ in round $t$, and refinements $n_S^t$ and $m_B^t$ in round $t$. $P_{t+1}$ has refinements $m_S^{t+1}$ and $n_B^{t+1}$ in round $t+1$, and refinements $n_S^t$ and $n_B^t$ in round $t$. Notice that $n_S^t$ and $n_B^{t+1}$ are identities and $n_B^t(d_S,d_B,d'_S)= m_B^{t}(d_S,m_B^{t+1}(d'_S,d_B))$ . Therefore, for any $d_S$ and $d_B$, going through the refinement process of rounds $t,t+1$ in both protocols will get the same outcome. Thus $P_t$ and $P_{t+1}$ have the same utilities.
    \item For $P_t$ versus $P_{t+1}$ with even $t >0$, the argument is symmetric to the previous case. \qedhere
\end{itemize}
\end{proof}

\subsection{Non-alternating refinement vs generic protocols}
\label{sec:generic}
Now we define generic protocols. Consider a generic protocol $P'$ of $k$ rounds. Before the first round, both parties Sally and Bob get their inputs sampled, i.e. $\theta_S \sim \D_S$ and $\theta_B \sim \D_B$. They also have the public randomness $R^{pub}$ and private randomness $R^{priv}_S$ and $R^{priv}_B$ sampled. In each round $t \in [k]$, 
\begin{itemize}
\item Sally sends message $g_S^t = g_S^t(h_{t-1}, \theta_S, R^{pub}, R^{priv}_S)$,
\item and simultaneously Bob sends message $g_B^t = g_B^t(h_{t-1}, \theta_B, R^{pub}, R^{priv}_B)$.
\end{itemize}
Here the history $h_t$ of first $t$ rounds is defined as the concatenation of messages sent in first $t$ rounds: $h_t = (g_S^1, g_B^1, g_S^2, g_B^2,\ldots, g_S^t, g_B^t)$. We define $\D_S^{t,g}$ as the distribution of $\theta_S$ given Bob's information after first $t$ round including history $h_t$, Bob's input $\theta_B$, Bob's private randomness $R^{priv}_B$ and public randomness $R^{pub}$: 
\[
\D_S^{t,g} = \D_S \mid h_t, \theta_B,  R^{priv}_B,  R^{pub}.
\]
And similarly we define $\D_B^{t,g}$ as following:
\[
\D_B^{t,g} = \D_B \mid h_t, \theta_S,  R^{priv}_S,  R^{pub}.
\]

The utilities of Sally and Bob in the generic protocol after $t$ rounds is defined as
\[
\pi^{t,g}_i(\D_S, \D_B) = \E_{\theta_S \sim \D_S, \theta_B \sim \D_B, R^{priv}_S, R^{priv}_B, R^{pub}}[u_i(a_S^*(\theta_S; \D_B^{t,g} ), \theta_S, \theta_B)] \quad i \in \{B, S\}.
\]
And it's easy to check that it is related to the utilities we defined earlier in the following way:
\[
\pi^{t,g}_i(\D_S, \D_B)   = \E[ \pi_i^0 (\D_S^{t,g}, \D_B^{t,g})].
\]
\begin{lemma}
\label{lem:generic}
For the generic protocol $P'$ described above, there exists a information refinement protocol $P$ such that Sally and Bob have the same utilities in both protocols, i.e.
\[
\pi^{k,g}_i(\D_S, \D_B) =\pi^k_i(\D_S, \D_B)  \quad i \in \{B, S\}.
\]
\end{lemma}

\begin{proof}
For any generic protocol $P'$, we construct $P$ by defining its messages $m_B$ and $m_S$. For each round $t$, and for any $d_S,d'_S \in \Delta(\Theta_S)$, $d_B,d'_B \in \Delta(\Theta_B)$ we set
\[
\Pr[m^{k-t+1}_i(d_S,d_B) = d'_i] = \Pr[\D_i^{t,g} = d'_i \mid \D_S^{t-1,g} = d_S, \D_B^{t-1,g} = d_B]\quad i \in \{B, S\}.
\]
Notice that the index $k-t+1$ comes from the fact that the information refinement protocol is defined in the reversed order of time.

We first check each step of $P$ is an information refinement, i.e. $\sum_{d'_i \in \Delta(\Theta_i)}  \Pr[m^{k-t+1}_i(d_S,d_B) = d'_i] \cdot d'_i = d_i$. For notational convenience, set $j \in \{B,S\} \setminus \{i\}$. Notice that $D_i^{t,g}$ equals to the distribution $\D_i$ conditioned on $h_t, R^{pub}, R_j^{priv}$, i,e. $\D_i^{t,g} = \D_i \mid h_t, R^{pub}, R_j^{priv}$. We have 
\begin{align*}
&\sum_{d'_i \in \Delta(\Theta_i)}  \Pr[m^{k-t+1}_i(d_S,d_B) = d'_i] \cdot d'_i \\
=& \sum_{d'_i \in \Delta(\Theta_i)} \Pr[\D_i^{t,g} = d'_i \mid \D_S^{t-1,g} = d_S, \D_B^{t-1,g} = d_B] \cdot d'_i \\
=& \sum_{h_t, R^{pub}, R_j^{priv}} \Pr[h_t, R^{pub}, R_j^{priv} \mid \D_S^{t-1,g} = d_S, \D_B^{t-1,g} = d_B] \cdot (D_i \mid h_t, R^{pub}, R_j^{priv})   \\
=& \E[D_i \mid \D_S^{t-1,g} = d_S, \D_B^{t-1,g} = d_B]\\
=& d_i
\end{align*}

By how we define the messages of $P$, it's easy to show by induction on $t = 1,\ldots,k$ that for any  $d_S \in \Delta(\Theta_S)$, $d_B \in \Delta(\Theta_B)$,
\[
\Pr[\D_S^{k-t} = d_S, \D_B^{k-t} = d_B] = \Pr[\D_S^{t,g} = d_S, \D_B^{t,g} = d_B].
\]

Finally, since
\[
\pi^{k,g}_i(\D_S, \D_B) = \E[ \pi_i^0 (\D_S^{k,g}, \D_B^{k,g})]
\]
and 
\[
\pi^{k}_i(\D_S, \D_B) = \E[ \pi_i^0 (\D_S^{0}, \D_B^{0})].
\]
We know $\pi^{k,g}_i(\D_S, \D_B)$ and $\pi^{k}_i(\D_S, \D_B)$ are equal.
\end{proof}

\begin{lemma}
For any information refine protocol $P$, there exists a generic protocol $P'$ with the same utilities.
\end{lemma}

\begin{proof}
We construct the generic protocol $P'$ from $P$ by defining its message functions $g^t_S$ and $g^t_B$. $P'$ does not use public randomness and it only uses private randomness. 

For $i \in \{B,S\}$ and any round $t = 1,\ldots,k$, message $g_i^t$ is in the space $\Delta(\Theta_i)$ and is set such that for any $d'_i \in \Delta(\Theta_i)$,
\[ \Pr_{R^{priv}_i} [g_i^t(h_{t-1},\theta_i, R_i^{priv}) = d_i'] = \frac{\Pr[m^{k-t+1}_i(g^{t-1}_S,g^{t-1}_B) = d'_i] \cdot d'_i(\theta_i)}{\sum_{d''_i \in \Delta(\Theta_i)} \Pr[m^{k-t+1}_i(g^{t-1}_S,g^{t-1}_B) = d''_i] \cdot d''_i(\theta_i)}.\]
Here $g_i^0$ is set to be $\D_i$ for notational convenience and $d'_i(\theta_i)$ means the probability of $\theta_i$ in distribution $d'_i$. With this definition, it's easy to check that $\D_i^{t,g} = g_i^t$. Then we get that for each round $t$, and for any $d_S,d'_S \in \Delta(\Theta_S)$, $d_B,d'_B \in \Delta(\Theta_B)$ we set
\[
\Pr[m^{k-t+1}_i(d_S,d_B) = d'_i] = \Pr[\D_i^{t,g} = d'_i \mid \D_S^{t-1,g} = d_S, \D_B^{t-1,g} = d_B]\quad i \in \{B, S\}.
\]
Notice this is exactly the same as what we set in the beginning of the proof of Lemma \ref{lem:generic}. The equivalence between utilities will just follow from the same proof as in Lemma \ref{lem:generic}.
\end{proof}

\section{Voluntary communication with binary types}\label{sec:voluntary_discussion}

We show in this section that the voluntary communication requirement is redundant in the equilibrium definition when both parties have binary type spaces ($\abs{\Theta_B} = \abs{\Theta_S} = 2$). This is a simple corollary of the following theorem: Since Bob's utility is convex in Sally's probability $q$, Sally's refinement cannot hurt Bob. 

\begin{theorem}
\label{thm:no_vol}
If the base game has $\pi^0_B(p,q)$ convex in $q$ and $\pi^0_S(p,q)$ convex in $p$. then $\pi^t_B(p,q)$ is convex in $q$ and $\pi^t_S(p,q)$ is convex in $p$, for any $t > 0$.
\end{theorem}

The proof will use the following lemma:
\begin{lemma}
\label{lem:convex}
Consider two functions $f,g : [0,1]\rightarrow \mathbb{R}$, and assume $g$ is convex. Define $m: [0,1] \rightarrow \Delta([0,1])$ such that for each $p \in [0,1]$, $m(p)$ maximizes $\E_{q\sim d}[f(q)]$ over all $d \in \Delta([0,1])$ and $\mu(d) = p$. If there are multiple $d$'s satisfying the above, $m(p)$ breaks ties in favor of maximizing $\E_{q\sim d}[g(q)]$. Formally:
$$m(p) \in A(p) =  \arg\max_{d: d \in \Delta([0,1]),\mu(d) = p} \E_{q\sim d}[f(q)]$$
$$m(p) \in B(p) = \arg\max_{d: d \in A(p)} \E_{q\sim d}[g(q)]$$
Then $G(p) =\E_{q\sim m(p)}[g(q)]$ is convex for $p\in[0,1]$.
\end{lemma}

\begin{proof}
Similar to Section \ref{sec:pi_1}, let $f^*$ be the concave hull of $f$, or in other words,
$f^*(p) =\max_{d: d \in \Delta([0,1]),\mu(d) = p} \E_{q\sim d}[f(q)]$.
Let $0 = p_0 < p_1 < \cdots < p_n = 1$ be the points that partition $f^*$ into piecewise linear functions. Notice that the proof is using the fact that this partition is in the 1-dimensional space. As you will see in Example \ref{ex:notconvex}, the proof would not work if the partition is in higher dimensional space. 
Define $\xi:[0,1] \rightarrow \Delta([0,1])$ to be the following:
\begin{itemize}
    \item If $p = p_i$ for some $i = 0,\ldots ,n$, $\xi(p)$ is a singleton distribution on $p_i$.
    \item If $p \in (p_i,p_{i+1})$ for some $i = 0,\ldots ,n-1$, $\xi(p)$ has probability density $(p_{i+1} - p) / (p_{i+1}-p_i)$ on $p_i$ and probability density $(p - p_i) / (p_{i+1}-p_i)$ on $p_{i+1}$.
\end{itemize}
It's easy to check that $\mu(\xi(p)) = p$. Now we want to show that $\xi(p) \in B(p)$ for $p \in [0,1]$.

\begin{itemize}
\item For $p = p_i$ for some $i = 0,\ldots ,n$, by the definition of the concave hull, we know $\xi(p)$ is the unique element in $A(p)$. Therefore, $\xi(p) \in B(p)$.

\item For $p \in (p_i,p_{i+1})$, by definition of the concave hull, we know that $\xi(p) \in A(p)$. Moreover, for any $d \in A(p)$, $d$ is supported on $[p_i,p_{i+1}]$. By the convexity of $g$, we know $\xi(p)$ maximizes $\E_{q\sim d}[g(q)]$ over all $d$ supported on $[p_i, p_{i+1}]$ and has $\mu(d) = p$. Therefore, $\xi(p) \in B(p)$.
\end{itemize}

Now we have $\xi(p) \in B(p)$. Notice that for any $d \in B(p)$, $\E_{q\sim d}[g(q)]$ are the same and $G(p)$ is uniquely determined. 
Therefore, we have $$G(p) = \E_{q\sim \xi(p)}[g(q)] = \frac{1}{p_{i+1}-p_i}(g(p_i) \cdot \left(p_{i+1}-p) + g(p_{i+1}) \cdot (p - p_i)\right) ~~~\forall p \in [p_i,p_{i+1}].$$

So $G$ is a piecewise linearization of $g$ and the convexity is preserved.
\end{proof}

\begin{proof}[Proof of Theorem \ref{thm:no_vol}]
We only prove the theorem for $\pi_B^t$, the argument is symmetric for $\pi^t_S$. We prove by induction. Suppose that $\pi^t_B(p,q)$ is convex in $q$, we want to show $\pi^{t+1}_B(p,q)$ is convex in $q$. 

If $t+1$ is a round by Bob sending the refinement, we have for any $p ,q,q_1,q_2 \in [0,1]$ with $q_1 + q_2 = 2q$,
\begin{align*}
\pi^{t+1}_B(p,q) &= \max_{d \in \Delta([0,1]), \mu(d)=p} \E_{r \sim d}[\pi^{t}_B (r,q)] \\
&\leq \frac{1}{2}\cdot \max_{d \in \Delta([0,1]), \mu(d)=p} \E_{r \sim d}[\pi^{t}_B (r,q_1) + \pi^{t}_B (r,q_2)] \\
&\leq \frac{1}{2}\cdot \left(\left(\max_{d \in \Delta([0,1]), \mu(d)=p} \E_{r \sim d}[\pi^{t}_B (r,q_1)]\right) +\left(\max_{d \in \Delta([0,1]), \mu(d)=p} \E_{r \sim d}[\pi^{t}_B (r,q_2)]\right) \right) \\
&= \frac{1}{2}\cdot \left( \pi^{t+1}_B(p,q_1) +\pi^{t+1}_B(p,q_2) \right).
\end{align*}
And this implies $\pi^{t+1}_B(p,q)$ is convex in $q$.

If $t+1$ is a round by Sally sending the refinement, we use Lemma \ref{lem:convex}, and set $f = \pi^t_S(p, \cdot), g = \pi^t_B(p,\cdot)$ for any $p \in [0,1]$. And we know $G = \pi^{t+1}_B(p,\cdot)$ and therefore $\pi^{t+1}_B(p,q)$ is convex in $q$.
\end{proof}

We show in the following example that Lemma \ref{lem:convex} does not hold if $f,g$ are in higher dimensional space (e.g. $f,g : [0,1]^2\rightarrow \mathbb{R}$). Therefore, the proof technique for Theorem \ref{thm:no_vol} would not work beyond binary types.

\begin{example}
\label{ex:notconvex}
Define  $f,g : [0,1]^2\rightarrow \mathbb{R}$ to be the following:
\begin{itemize}
    \item $f = 0$ except $f(0,1/2)=f(1,1/2)=1$, $f(1/2,3/4) = f(1/2,1/4)=0.9$.
    \item $g(p_1,p_2) = (p_1-1/2)^2 + (p_2-1/2)^2$. $g$ is convex.
\end{itemize}
Define $m: [0,1]^2 \rightarrow \Delta([0,1]^2)$ to satisfy the followings for each $p=(p_1,p_2) \in [0,1]^2$:
\begin{itemize}
    \item  $m(p)$ maximizes $\E_{q\sim d}[f(q)]$ over all $d \in \Delta([0,1]^2)$ and $\mu(d) = p$.
    \item If there are multiple $d$'s satisfy the above, $m(p)$ breaks tie in favor of maximizing $\E_{q\sim d}[g(q)]$.
\end{itemize}
We will show $G(p) =\E_{q\sim m(p)}[g(q)]$ is not convex by showing that
$G(1/2,1/2) > \frac{1}{2}(G(1/2,3/8) + G(1/2,5/8))$.

By the definition of $f$, we know $m(1/2,1/2)$ will be $(0,1/2)$ w.p. $1/2$ and $(1,1/2)$ w.p. $1/2$. Therefore $G(1/2,1/2) = \frac{1}{2} ( g(0,1/2) + g(1,1/2)) = 1/4$.

Again by the definition of $f$, we know $m(1/2,3/8)$ will be $(0,1/2)$ w.p. $1/4$, $(1,1/2)$ w.p. $1/4$ and $(1/2,3/4)$ w.p. $1/2$. Therefore $G(1/2,3/8) = g(0,1/2)/4 + g(1,1/2)/4 + g(1/2,3/4) / 2 = 1/16 + 1/16 + 1/32 = 5/32$. By symmetry, we have $G(1/2,5/8) = G(1/2,3/8) = 5/32$. 

Therefore $G(1/2,1/2) > \frac{1}{2}(G(1/2,3/8) + G(1/2,5/8))$, and $G$ is not convex.
\end{example}

\section{One Message per Player Is Not Enough}\label{sec:three_rounds}

If the buyer has a binary type space, we showed that $2$ rounds of communication are enough to achieve efficiency. We now show an example where strictly better efficiency can be achieved with more rounds. 

Consider a setting where $\Theta_B = \{v_1, v_2, v_3\}$ and $\Theta_S = \{c_1, c_2\}$. With $c_1 < c_2 < v_1 < v_2 < v_3$. Our first step is a structural characterization of Bob's message in the first round. We will show that it is without loss of generality to consider that Bob will always use one of $12$ messages.

\begin{lemma}\label{lemma:twelve_distr}
Given state $(\D_B, \D_S)$ and an information refinement $M $ of $\D_B$ then there is an information refinement $M'$ of $\D_B$ such that:
$$\E_{\D'_B \sim M'} [\pi_i^0(\D_S, \D'_B)] \geq \E_{\D'_B \sim M} [\pi_i^0(\D_S, \D'_B)], \text{ for } i \in \{B,S\}$$
and the distributions $\D'_B$ in the support of $M'$ consists of one of the following $12$ possibilities:
\begin{enumerate}
    \item Distribution of support size $1$ ($3$ possibilities)
    \item Distribution of support size $2$ where Sally is indifferent between pricing at either when her cost is $c_1$. ($3$ possibilities)
    \item Distribution of support size $2$ where Sally is indifferent between pricing at either when her cost is $c_2$. ($3$ possibilities)
    \item Distribution of support size $3$ where Sally is indifferent between pricing at either when her cost is $c_1$.
    \item Distribution of support size $3$ where Sally is indifferent between pricing at either when her cost is $c_2$.
    \item Distribution of support size $3$ where Sally is indifferent between pricing at $v_1$ and $v_2$ when her cost is $c_1$ and is indifferent between pricing at $v_2$ and $v_3$ when her cost is $c_2$.
\end{enumerate}
\end{lemma}

\begin{proof}
Given any distribution $\D'_B$ we will show how to decompose it into distributions like the ones in the statement of the lemma such that both Sally and Bob weakly improve their payoffs.

\noindent \emph{Case 1:} if $\D'_B$ has support size $1$ then it is already in the desired form. 

\noindent \emph{Case 2:} if $\D'_B$ has support $\{v_L, v_H\}$ and the types have probability $p_L$ and $p_H$ respectively, then assume Sally is not indifferent at any of her cost, then:
\begin{equation}\label{eq:diff_prices}
p_H (v_H - c_0) \neq (v_L - c_0) \qquad 
  p_H (v_H - c_1) \neq (v_L - c_1)
\end{equation}
Then define for each $\epsilon > 0$ define two distributions:
\begin{itemize}
    \item $\D''_B$ puts all the mass on the low type
    \item $\D'''_B$ puts mass $p_H / (1-\epsilon)$ on the high type and otherwise on the low type.
\end{itemize}
Now instead replace the message $\D'_B$ by message $\D''_B$  with probability $\epsilon$ and message $\D'''_B$  with probability $1-\epsilon$. As we keep increasing $\epsilon$ either one of the equations \eqref{eq:diff_prices} holds with equality in which case $\D'''_B$ is in the format of items $2$ and $3$ in the statement or $\epsilon$ reaches $1$ in which case $\D'''_B$ has support size $1$. In either case, the only change is that in cases where Sally wasn't selling before, she is now selling at price $v_L$. Hence Bob's utility remains the same and Sally's utility can only improve.

\noindent \emph{Case 3:} if $\D'_B$ has full support $\{v_1, v_2, v_3\}$ and Sally prices at $v_1$ when her cost is $c_2$: If Sally is indifferent at all three prices when her cost is $c_2$, $\D'_B$ is in the form of item 5 in the statement. Otherwise, we decompose $\D'_B$ into $\D''_B$ and $\D'''_B$. In $\D''_B$, Sally is indifferent at all three prices at cost $c_2$. The probability of sending $\D''_B$ is calculated so that one of the values is exhausted, so that $\D'''_B$ has support size of $1$ or $2$. In $\D'''_B$, Sally's optimal price must be $v_1$ at cost $c_2$, since she prefers $v_1$ in $\D'_B$ and she is indifferent in $\D''_B$. In $\D'_B$, $\D''_B$, and $\D'''_B$, Sally must price at $v_1$ when her cost is $c_1$, as the cost reduction also benefits the price of $v_1$ the most. Therefore, this decomposition does not change Sally's strategy, and we have reduced the value distribution to item 5 in the statement and case 1 or 2 in this proof.

\noindent \emph{Case 4:} if $\D'_B$ has full support $\{v_1, v_2, v_3\}$ and Sally prices at $v_3$ when her cost is $c_2$: We decompose $\D'_B$ into $\D''_B$ and $\D'''_B$, where in $\D''_B$, Sally is indifferent at all three prices at $c_1$, and $\D'''_B$ has support size of $1$ or $2$. When Sally has cost $c_1$, her optimal price for $\D'''_B$ is the same as that for $\D'_B$, and her optimal price for $\D''_B$ is $v_1$. When she has cost $c_2$, she sets price at $v_3$ for $\D'_B$. Therefore, this decomposition makes the price Sally sets weakly decrease, which weakly benefits Bob and the sum of Sally's and Bob's utilities. We have thus reduced the value distribution to item 4 in the statement and case 1 or 2 in this proof.

\noindent \emph{Case 5:} if $\D'_B$ has full support $\{v_1, v_2, v_3\}$ and Sally prices at $v_2$ when her cost is $c_2$: We decompose $\D'_B$ into $\D''_B$ and $\D'''_B$, where $\D''_B$ is in the form of item 6 in the statement and $\D'''_B$ has support size of at most $2$. When Sally's cost is $c_2$, the optimal price for $\D''_B$ and $\D'''_B$ are still $v_2$. When her cost is $c_1$, the optimal price for $\D'''_B$ must be at most $v_2$, since cost reduction benefits lower prices more, and thus the optimal price for $\D'''_B$ must be the same as that for $\D'_B$, since in $\D''_B$, Sally is indifferent between prices of $v_1$ and $v_2$. Therefore, similar to case 4, this decomposition makes Sally's price weakly decrease, and we have reduced the value distribution to item 6 in the statement and case 1 or 2 in this proof.
\end{proof}

A corollary of the previous lemma is that it is possible to compute Bob's best response by solving a linear program. Let $p^i = [p^i_1, p^1_2, p^i_3]$ for $i = 1, \ldots, 12$ be the probabilities associated with the distributions in the statement of Lemma \ref{lemma:twelve_distr}. We are given a pair $(\D_S, \D_B)$ where $\D_S$ is represented by the probability $q$ that Sally has the high type and $\D_B$ is represented by a vector of probabilities $p = [p_1, p_2, p_3]$. Now consider the following pair of programs:

$$\left.
\begin{aligned}
\pi^1_B(q,p) = & \max \sum_{i=1}^{12} w_i \pi_B^0(q, p^i)  \\
& \begin{aligned}
\text{ s.t.} \text{ } &  \sum_{i=1}^{12} w_i \pi_S^0(q, p^i) \geq 0 \\
& \sum_{i=1}^{12} w_i p^i = p \\
& \sum_{i=1}^{12} w_i = 1 \\
& w_i \geq 0, \ i=1, \ldots, 12
\end{aligned}
\end{aligned}
\quad \right. \left| \quad
\begin{aligned}
\pi^1_S(q,p) = & \max \sum_{i=1}^{12} w_i \pi_S^0(q, p^i)  \\
& \begin{aligned}
\text{ s.t.} \text{ } & \sum_{i=1}^{12} w_i \pi_B^0(q, p^i) = \pi_B^1(q, p) \\
& \sum_{i=1}^{12} w_i p^i = p \\
& \sum_{i=1}^{12} w_i = 1 \\
& w_i \geq 0, \ i=1, \ldots, 12
\end{aligned}
\end{aligned}
\right. $$

The first program computes the payoff after Bob's best response in the first round. The second program computes Sally's payoffs after Bob's response by finding the most beneficial tie-breaking for Sally.

\subsection{An Example requiring $3$ rounds of communication}

Using Lemma \ref{lemma:twelve_distr} and the linear programming formulation, we can now show an example where $3$ rounds of communication are required to achieve efficiency. Let Bob's types be $\Theta_B = \{3, 6, 12\}$ with probabilities $p=[1/3,1/3,1/3]$ and Sally's types be $\Theta_S = \{0, 2\}$ with $q = 1/5$ probability on the high type. 

Without any communication, Sally sets price at $6$ when her cost is $0$ and at $12$ when her cost is $2$, giving:
\[
\pi^0_S(q,p) = \frac{4}{5} \cdot 4 + \frac{1}{5} \cdot \frac{10}{3} = \frac{58}{15}, \quad \pi^0_B(q,p) = \frac{4}{5} \cdot 2 = \frac{8}{5}, \quad W^0(q,p) = \frac{82}{15} < W^*(q,p) = \frac{33}{5}
\]

Using the linear program above, we can compute the payoffs in round $1$ if for the original vector of probabilities $p$ for Bob and any probabiliy $q'$ of Sally having the high type: 
\[
\pi^1_S(q',p) =
\begin{cases}
    4 - \frac{2}{3} q'       & \quad \text{if } q' < \frac{1}{3}\\
    \frac{42}{9} - \frac{4}{3} q'  & \quad \text{if } \frac{1}{3} \leq q' < \frac{2}{3}\\
    \frac{16}{3} - 2q'  & \quad \text{if } q' \geq \frac{2}{3}\\
\end{cases}, \qquad 
\pi^1_B(q',p) =
\begin{cases}
    3 - 3q'       & \quad \text{if } q' < \frac{1}{3}\\
    \frac{7}{3} - q'  & \quad \text{if } \frac{1}{3} \leq q' < \frac{2}{3}\\
    \frac{5}{3}  & \quad \text{if } q' \geq \frac{2}{3}\\
\end{cases}
\]

In particular, for the original value of $q=1/5$ we have:
\[
\pi^1_S(q,p) = \frac{58}{15}, \quad \pi^1_B(q,p) = \frac{12}{5}, \quad W^1(q,p) = \frac{94}{15} < W^*(q,p) = \frac{33}{5}
\]

Sally's best response at $t = 2$ is to refine Bob's prior to $q = 0$ with probability 2/5 and $q = 1/3$ with probability 3/5, hence:
\[
\pi^2_i(1/5, p) = \frac{3}{5} \pi^1_i(1/3, p) + \frac{2}{5} \pi^1_i(0,p)
\]

Substituting the numerical values, we get:
\[
\pi^2_S(q,p) = \frac{3}{5} \cdot \frac{38}{9} + \frac{2}{5} \cdot 4 = \frac{62}{15}, \quad \pi^2_B(q,p) = \frac{3}{5} \cdot 2 + \frac{2}{5} \cdot 3 = \frac{12}{5}, \quad W^2(q,p) = \frac{98}{15} < W^*(q,p) = \frac{33}{5}
\]

By Theorem \ref{thm:main_trade_thm}, since $\pi_S^2(q,p) + \pi_B^2(q,p) < W^*(p,q)$ there must be some $t>2$ such that $\pi_S^t(q,p) + \pi_B^t(q,p) > \pi_S^2(q,p) + \pi_B^2(q,p)$. Hence the message complexity must be at least $3$.

\end{document}